\documentclass[twocolumn,journal]{IEEEtran}
\usepackage[T1]{fontenc}
\usepackage[latin9]{inputenc}
\usepackage{color}
\usepackage{array}
\usepackage{multirow}
\usepackage{dsfont}
\usepackage{amsmath}
\usepackage{amsthm}
\usepackage{amssymb}
\usepackage{graphicx}
\usepackage[unicode=true,
 bookmarks=true,bookmarksnumbered=true,bookmarksopen=true,bookmarksopenlevel=1,
 breaklinks=false,pdfborder={0 0 0},pdfborderstyle={},backref=false,colorlinks=false]
 {hyperref}
\hypersetup{pdftitle={Your Title},
 pdfauthor={Your Name},
 pdfpagelayout=OneColumn, pdfnewwindow=true, pdfstartview=XYZ, plainpages=false}

\makeatletter

\providecommand{\tabularnewline}{\\}

\theoremstyle{definition}
\newtheorem{defn}{\protect\definitionname}
\theoremstyle{definition}
\newtheorem{example}{\protect\examplename}
\theoremstyle{plain}
\newtheorem{prop}{\protect\propositionname}
\theoremstyle{plain}
\newtheorem{lem}{\protect\lemmaname}
\theoremstyle{plain}
\newtheorem{thm}{\protect\theoremname}
\ifx\proof\undefined\
  \newenvironment{proof}[1][\proofname]{\par
    \normalfont\topsep6\p@\@plus6\p@\relax
    \trivlist
    \itemindent\parindent
    \item[\hskip\labelsep
          \scshape
      #1]\ignorespaces
  }{%
    \endtrivlist\@endpefalse
  }
  \providecommand{\proofname}{Proof}
\fi
\theoremstyle{remark}
\newtheorem{rem}{\protect\remarkname}
\theoremstyle{plain}
\newtheorem{cor}{\protect\corollaryname}

\usepackage{epstopdf}

\usepackage{tikz}
\usetikzlibrary{arrows}
\usetikzlibrary{shadows}
\usetikzlibrary{positioning}

\def\squarecorner#1{
    \pgf@x=\the\wd\pgfnodeparttextbox%
    \pgfmathsetlength\pgf@xc{\pgfkeysvalueof{/pgf/inner xsep}}%
    \advance\pgf@x by 2\pgf@xc%
    \pgfmathsetlength\pgf@xb{\pgfkeysvalueof{/pgf/minimum width}}%
    \ifdim\pgf@x<\pgf@xb%
        \pgf@x=\pgf@xb%
    \fi%
    \pgf@y=\ht\pgfnodeparttextbox%
    \advance\pgf@y by\dp\pgfnodeparttextbox%
    \pgfmathsetlength\pgf@yc{\pgfkeysvalueof{/pgf/inner ysep}}%
    \advance\pgf@y by 2\pgf@yc%
    \pgfmathsetlength\pgf@yb{\pgfkeysvalueof{/pgf/minimum height}}%
    \ifdim\pgf@y<\pgf@yb%
        \pgf@y=\pgf@yb%
    \fi%
    \ifdim\pgf@x<\pgf@y%
        \pgf@x=\pgf@y%
    \else
        \pgf@y=\pgf@x%
    \fi
    \pgf@x=#1.5\pgf@x%
    \advance\pgf@x by.5\wd\pgfnodeparttextbox%
    \pgfmathsetlength\pgf@xa{\pgfkeysvalueof{/pgf/outer xsep}}%
    \advance\pgf@x by#1\pgf@xa%
    \pgf@y=#1.5\pgf@y%
    \advance\pgf@y by-.5\dp\pgfnodeparttextbox%
    \advance\pgf@y by.5\ht\pgfnodeparttextbox%
    \pgfmathsetlength\pgf@ya{\pgfkeysvalueof{/pgf/outer ysep}}%
    \advance\pgf@y by#1\pgf@ya%
}
\makeatother

\pgfdeclareshape{square}{
    \savedanchor\northeast{\squarecorner{}}
    \savedanchor\southwest{\squarecorner{-}}

    \foreach \x in {east,west} \foreach \y in {north,mid,base,south} {
        \inheritanchor[from=rectangle]{\y\space\x}
    }
    \foreach \x in {east,west,north,mid,base,south,center,text} {
        \inheritanchor[from=rectangle]{\x}
    }
    \inheritanchorborder[from=rectangle]
    \inheritbackgroundpath[from=rectangle]
}

\usetikzlibrary{shapes}
\usetikzlibrary{arrows,decorations.pathmorphing}

\usepackage{hyperref}
\usepackage{pdfsync}

\usepackage{mathtools,amssymb,lipsum}

\usepackage{cuted}
\usepackage{bbm}

\usepackage{algorithm,algpseudocode}
\usepackage{mathdots}  
\usepackage{lipsum}
\usepackage{mathtools}
\usepackage{cuted}

\def\tablename{Table}
\usepackage{subfig}

\makeatletter

\newcommand{\Rmnum}[1]{\expandafter\@slowromancap\romannumeral #1@}
\makeatother

\makeatother

\providecommand{\corollaryname}{Corollary}
\providecommand{\definitionname}{Definition}
\providecommand{\examplename}{Example}
\providecommand{\lemmaname}{Lemma}
\providecommand{\propositionname}{Proposition}
\providecommand{\remarkname}{Remark}
\providecommand{\theoremname}{Theorem}

\begin{document}
\title{Distributed Stabilization of Signed Networks via Self-loop Compensation}
\author{Haibin~Shao,~\IEEEmembership{Member,~IEEE,} and Lulu Pan, \IEEEmembership{Member,~IEEE}\thanks{This work was supported by the National Science Foundation of China
(Grant No.62103278, 61973214,  61963030).\textcolor{black}{{} }}\thanks{H.\,Shao and L.\,Pan are with the Department of Automation, Shanghai
Jiao Tong University, and Key Laboratory of System Control and Information
Processing, Ministry of Education of China, and Shanghai Engineering
Research Center of Intelligent Control and Management, Shanghai, 200240,
China. }\thanks{\textcolor{black}{Corr}esponding author: Lulu Pan (llpan@sjtu.edu.cn).}}
\maketitle
\begin{abstract}
This paper examines the stability and distributed stabilization of
signed multi-agent networks. Here, positive semidefiniteness is not
inherent for signed Laplacians, which renders the stability and consensus
of this category of networks intricate. First, we examine the stability
of signed networks by introducing a novel graph-theoretic objective
\emph{negative cut set}, which implies that manipulating negative
edge weights cannot change a unstable network into a stable one. Then,
inspired by the diagonal dominance and stability of matrices, a local
state damping mechanism is introduced using \emph{self-loop compensation}.
The self-loop compensation is only active for those agents who are
incident to negative edges and can stabilize signed networks in a
fully distributed manner. Quantitative connections between self-loop
compensation and the stability of the compensated signed network are
established for a tradeoff between compensation efforts and network
stability. Necessary and/or sufficient conditions for predictable
cluster consensus of compensated signed networks are provided. The
optimality of self-loop compensation is discussed. Furthermore, we
extend our results to directed signed networks where the symmetry
of signed Laplacian is not free. The correlation between the stability
of the compensated dynamics obtained by self-loop compensation and
eventually positivity is further discussed. Novel insights into the
stability of multi-agent systems on signed networks in terms of self-loop
compensation are offered. Simulation examples are provided to demonstrate
the theoretical results. 
\end{abstract}

\begin{IEEEkeywords}
Signed Laplacian; self-loop compensation; positive semidefinite; positive
feedback; structural balance; virtual leader.
\end{IEEEkeywords}

\section{Introduction}

Achieving consensus via diffusive local interactions is an important
paradigm in distributed algorithms on networks, where graph Laplacian
plays a central role in analysis and design of the network performance\ \cite{mesbahi2010graph,kia2019tutorial,qin2016recent}.
A neat feature of graph Laplacian lies in its positive semidefiniteness
which determines the stability and consensus of multi-agent networks
\cite{olfati2004consensus,godsil2001algebraic}. 

For a long time, the consensus problem has been extensively examined
in networks with only cooperative interactions\ \cite{mesbahi2010graph}.
However, antagonistic interactions can also exist which can capture,
for instance, the antagonism in social networks which can lead to
opinion polarization, or critical lines in power networks which may
cause small-disturbance instability \cite{song2019extension,song2017network,ding2016negative,proskurnikov2015opinion}.
 Recently, there has been a growing interest in multi-agent systems
on signed networks \cite{meng2018uniform,chen2020spectral,altafini2013consensus,xia2011clustering,shi2019dynamics,pan2020sciencechina,pan2018bipartite,Pan2021Tac,wang2020characterizing,cheng2020seeking,shi2021sub}.
In this line of work, although the optimal selection of edge weights
for a consensus algorithm can be involving negative weights \cite{xiao2004fast},
a notable observation is that the consensus or even stability of the
multi-agent network cannot be guaranteed by negating the weight of
only one edge from positive to negative \cite{zelazo2017robustness,pan2016laplacian}. 

From a \emph{microscopic} view, the negative edges in diffusively
coupled networks turn out to be a mechanism of \emph{positive feedback},
that is, the state deviations amongst neighboring agents reinforce;
while the positive edges are a mechanism of \emph{negative feedback},
namely, the state deviations between neighboring agents reduce \cite{shi2019dynamics}.
Notably, the positive and negative feedbacks are also ubiquitously
observed in biological networks functioning as mechanisms of activation
and inhibition  \cite{sepulchre2019control,markowetz2007inferring,wong2004combining,albert2005scale,barabasi2011network}.
From a \emph{macroscopic} view, the stability guarantee when implementing
the diffusive interaction protocol on signed networks turns out to
be a critical prerequisite with respect to the network functionality
and performance \cite{zelazo2014definiteness,meng2018convergence}.
Moreover, the typical collective behaviors of signed networks can
often be cluster consensus whose explicit characterization is also
challenging \cite{pan2016laplacian,pan2021cluster,xia2011clustering,zelazo2017robustness}.
In fact, the stability and cluster consensus problems of signed networks
can be attributed to examining the spectral properties of signed Laplacian,
i.e., the Laplacian matrix of signed networks \cite{hershkowitz1992recent,pan2016laplacian,bronski2014spectral,chen2017spectral}. 

The diagonal dominance of graph Laplacian (a global property) naturally
guarantees the stability of the multi-agent system on unsigned networks
\cite{pan2016laplacian,altafini2013consensus}. For signed networks,
however, there can exist mutually cancellation amongst positive and
negative weights when summing them together for the diagonal entries
of signed Laplacian, whose inherent properties under unsigned networks
such as diagonal dominance and positive semidefiniteness of Laplacian
are no longer valid. Therefore, the analysis of signed Laplacian cannot
fall into the traditional analysis framework such as $M$-matrix theory
or Gershgorin disc analysis \cite{horn1990matrix,olfati2004consensus}. 

In \cite{zelazo2017robustness}, it is shown that a signed Laplacian
with only one negative edge is positive semidefinite if the magnitude
of the negative edge weight is less than or equal to the reciprocal
of the effective resistance between the nodes of the negative edge
over the positive subgraph.\textbf{ }This result was then extended
therein to the network with multiple negative weights under the restriction
that different negative edges are not on the same cycle. The aforementioned
result has been subsequently re-examined from the perspective of geometrical
and passivity-based approaches \cite{chen2016definiteness}. The spectral
properties of signed Laplacians with connections to eventual positivity
were recently examined in \cite{altafini2019investigating,chen2020spectral,altafini2014predictable}.
Optimal edge weight allocation for positive semidefiniteness of signed
Laplacians was investigated in \cite{wei2018optimal}. The stability
of the signed Laplacian matrix is also examined in terms of matrix
signatures \cite{pan2016laplacian,bronski2014spectral} and the Schur
stability criterion \cite{pirani2014stability}. 

Although notable results have been developed on the interplay between
the global characterization of edge weights and stability/consensus
of signed networks from the perspective of such as effective resistance
in electrical networks (only applicable for undirected networks) or
eventually positivity, the related results only provide the theoretical
guarantee on positive semidefiniteness of signed Laplacian using global
information, and this information often cannot be efficiently obtained
by each agent only using local observations. So far, to our best knowledge,
very few works appeal for measures that should be taken to stabilize
a signed network (undirected or directed) in a fully distributed manner,
which is exactly the pursuit of this work.

Matrix stability theory  plays a central role in characterizing stability
criteria of dynamical systems \cite{kaszkurewicz2012matrix,arcak2011diagonal,hershkowitz1992recent}.
A Hermitian diagonally dominant matrix with real non-negative diagonal
entries is positive semidefinite, which is a crucial fact in the stability
guarantee of Laplacian of unsigned networks. However, for signed networks,
a remarkable fact is that the magnitude of weights associated with
negative edges between neighboring agents may lay no influence on
shifting an unstable signed network with cut set into stable \cite{song2017characterization,zelazo2017robustness,pan2016laplacian}.
These facts motivate us to develop a suitably distributed paradigm
to guarantee the stability and consensus of signed networks. To regain
the diagonal dominance of signed Laplacian, the manipulation of its
diagonal entries turns out to be effective since it can be realized
by each agent properly adopting feedback gains on their state in the
respective protocol, namely, this process can be implemented in a
fully distributed manner. We shall refer to the aforementioned operation
as \emph{self-loop compensation} in the following discussions. The
self-loop compensation is essentially introducing a state damping
mechanism for agents that are incident to negative edges. However,
this damping mechanism can be time- and resource-consuming which are
undesired, especially for networks whose units are highly energy-constrained
(e.g., a swarm of micro drone). It is therefore necessary to examine
an efficient selection of self-loop compensation for stabilizing signed
networks.

The main contributions of this paper can be summarized as follows.

1) A graph-theoretic characterization of stability of signed Laplacian
is provided from the perspective of negative cut set\emph{ }which
not only provides insight into the design of stable signed networks
but also indicates that manipulation of edge weights cannot change
an unstable signed network containing negative cut set into a stable
one.

2) Compared with the existing results that provide global characterizations
of signed Laplacian's positive semi-definiteness, this paper examines
the self-loop compensation mechanism in the design of local interaction
protocol amongst agents and further explores the interplay between
self-loop compensation of signed Laplacian and stability/consensus
of the multi-agent system on signed networks. 

3) We show that only agents who are incident to negative edges need
to be compensated. Analytical connections between the self-loop compensation
and the collective behaviors of the compensated signed network are
established. Necessary and sufficient conditions for (cluster) consensus
of compensated signed networks are provided as well as the explicit
characterization of their respective steady-states. The optimality
of self-loop compensation is also discussed. 

4) Both undirected and directed signed networks are investigated.
It is shown that structurally imbalanced networks need less self-loop
compensation to be stable than structurally balanced ones. The correlation
between the stability of the compensated dynamics obtained by self-loop
compensation and eventually positivity is finally discussed.

The remainder of this paper is organized as follows. Notions and graph
theory are introduced in \S 2. We then provide the discussion of
Laplacian dynamics on signed networks from the perspective of positive
and negative feedback loops in \S 3. A graph-theoretic stability
analysis using negative cut set is provided in \S 4, followed by
the introduction of self-loop compensation and results concerning
the magnitude of self-loop compensation and stability and (cluster)
consensus of the compensated signed network in \S 5. Furthermore,
we discuss the correlation between the compensated signed Laplacian
and eventually positivity in \S 6. Concluding remarks are finally
provided in \S 7.

\section{Preliminaries}

\subsection{Notations}

Let $\mathbb{R}$, $\mathbb{C}$, and $\mathbb{N}$ be the set of
real, complex, and natural numbers, respectively. Denote $\underline{s}=\left\{ 1,\ldots,s\right\} $
for an $s\in\mathbb{N}$. Let $\mathfrak{D}_{\ge0}$ denote the class
of real diagonal matrices $D=(d_{ij})\in\mathbb{R}^{n\times n}$ such
that $d_{ii}\ge0$ for $i\in\underline{n}$ and $d_{ij}=0$ for $i\not=j$.
Let $\text{{\bf Re}}(\cdot)$ denote the real part of a complex number.
Denote the entry of a matrix $M\in\mathbb{R}^{n\times n}$ located
in $i$-th row and $j$-th column as $[M]_{ij}$. Denote the eigenvalues
of the matrix $M=[m_{ij}]\in\mathbb{R}^{n\times n}$ as $\lambda_{i}(M)$
where $i\in\underline{n}$. The spectrum of the matrix $M$ is the
set of all its eigenvalues, denoted by $\boldsymbol{\lambda}(M)$.
The spectral radius of a square matrix $M$ is the largest absolute
value of its eigenvalues, denoted by $\rho(M)$. Let $\mathds{1}_{n}$
and $\mathbf{0}_{n}$ denote $n\times1$ vector of all ones and zeros,
respectively. We write $M\succ0$ ($M\succeq0$) if a symmetric matrix
$M\in\mathbb{R}^{n\times n}$ is positive definite (semidefinite).
A matrix $M=[m_{ij}]\in\mathbb{R}^{p\times q}$ is positive (nonnegative),
then we denote $M>0$ ($M\geq0$). For a pair of matrices $M\in\mathbb{R}^{p\times q}$
and $M^{\prime}\in\mathbb{R}^{p\times q}$, we write $M\geq M^{\prime}$
($M>M^{\prime}$) if $M-M^{\prime}\ge0$ ($M-M^{\prime}>0$). The
absolute value of a matrix $M=[m_{ij}]\in\mathbb{R}^{p\times q}$
is denoted by $|M|=[|m_{ij}|]\in\mathbb{R}^{p\times q}$. For two
vectors $\boldsymbol{x},\boldsymbol{y}\in\mathbb{R}^{q}$, if there
exist $i^{\prime},i^{\prime\prime}\in\underline{q}$ ($i^{\prime}\text{\ensuremath{\not}=}i^{\prime\prime}$)
such that $[\boldsymbol{x}]_{i^{\prime}}>[\boldsymbol{y}]_{i^{\prime}}$
and $[\boldsymbol{x}]_{i^{\prime\prime}}<[\boldsymbol{y}]_{i^{\prime\prime}}$,
then we write $\boldsymbol{x}\veebar\boldsymbol{y}$. 

\subsection{Graph Theory}

Let $\mathcal{G}=(\mathcal{V},\mathcal{E},W)$ denote a network with
the node set $\mathcal{V}=\left\{ 1,2,\ldots,n\right\} $, the edge
set $\mathcal{E}\subseteq\mathcal{V}\times\mathcal{V}$, and the adjacency
matrix $W=[w_{ij}]\in\mathbb{R}^{n\times n}$. Here, adjacency matrix
$W$ satisfies $w_{ij}\not=0$ if and only if $(i,j)\in\mathcal{E}$
and $w_{ij}=0$ otherwise. A network is undirected when $(i,j)\in\mathcal{E}$
if and only if $(j,i)\in\mathcal{E}$; otherwise, it is directed.
A network $\mathcal{G}$ is referred to as signed network if there
exist an edge $(i,j)\in\mathcal{E}$ such that $w_{ij}<0$, otherwise,
$\mathcal{G}$ is an unsigned network. A network is connected if any
two distinct nodes are reachable from one another via paths. The neighbor
set of an agent $i\in\mathcal{V}$ is $\mathcal{N}_{i}=\left\{ j\in\mathcal{V}|(i,j)\in\mathcal{E}\right\} $.
The neighbor set of an agent $i\in\mathcal{V}$ can be divided by
$\mathcal{N}_{i}=\mathcal{N}_{i}^{+}\cup\mathcal{N}_{i}^{-}$ where
$\mathcal{N}_{i}^{+}=\left\{ j\in\mathcal{V}|(i,j)\in\mathcal{E}\thinspace\text{and}\thinspace w_{ij}>0\right\} $
and $\mathcal{N}_{i}^{-}=\left\{ j\in\mathcal{V}|(i,j)\in\mathcal{E}\thinspace\text{and}\thinspace w_{ij}<0\right\} .$
Notably, \emph{structurally balance} is an important structure of
signed networks \cite{harary1953notion,cartwright1956structural}.
A signed network $\mathcal{G}=(\mathcal{V},\mathcal{E},W)$ is structurally
balanced if there is a bipartition of the node set $\mathcal{V}$,
say $\mathcal{V}_{1}$ and $\mathcal{V}_{2}$ such that $\mathcal{V}=\mathcal{V}_{1}\cup\mathcal{V}_{2}$
and $\mathcal{V}_{1}\cap\mathcal{V}_{2}=\emptyset$, satisfies that
the weights on the edges within each subset is positive, but negative
for edges between the two subsets\ \cite{harary1953notion}. A signed
network is structurally imbalanced if it is not structurally balanced.

\section{Signed Laplacian Dynamics}

Consider the signed Laplacian $\mathcal{L}(\mathcal{G})=[l_{ij}]\in\mathbb{R}^{n\times n}$
where $l_{ij}=\sum_{k=1}^{n}w_{ik}$ for $i=j$ and $l_{ij}=-w_{ij}$
for $i\ne j$. Note that $\mathcal{L}(\mathcal{G})\mathds{1}_{n}=\mathbf{0}_{n}$
holds which implies the signed Laplacian always has a zero eigenvalue.
We know that the signed Laplacian $\mathcal{L}(\mathcal{G})$ implies
the following interaction protocol amongst agents in a local level,

\begin{equation}
\dot{x}_{i}(t)=\sum_{j\text{\ensuremath{\in}\ensuremath{\mathcal{N}}}_{i}}w_{ij}(x_{j}(t)-x_{i}(t)),\thinspace i\in\mathcal{V},\label{eq:consensus-protocol-local}
\end{equation}
where $x_{i}(t)\in\mathbb{R}$ denotes the state of agent $i\in\mathcal{V}$.
Moreover, the overall dynamics of (\ref{eq:consensus-protocol-local})
can be characterized by the following signed Laplacian dynamics, 
\begin{equation}
\dot{\boldsymbol{x}}(t)=-\mathcal{L}(\mathcal{G})\boldsymbol{x}(t),\label{eq:consensus-protocol-overall}
\end{equation}
where $\boldsymbol{x}(t)=[x_{1}(t),x_{2}(t),\ldots,x_{n}(t)]^{T}$.
\begin{figure}[htbp]
\centering{}\includegraphics[width=9cm]{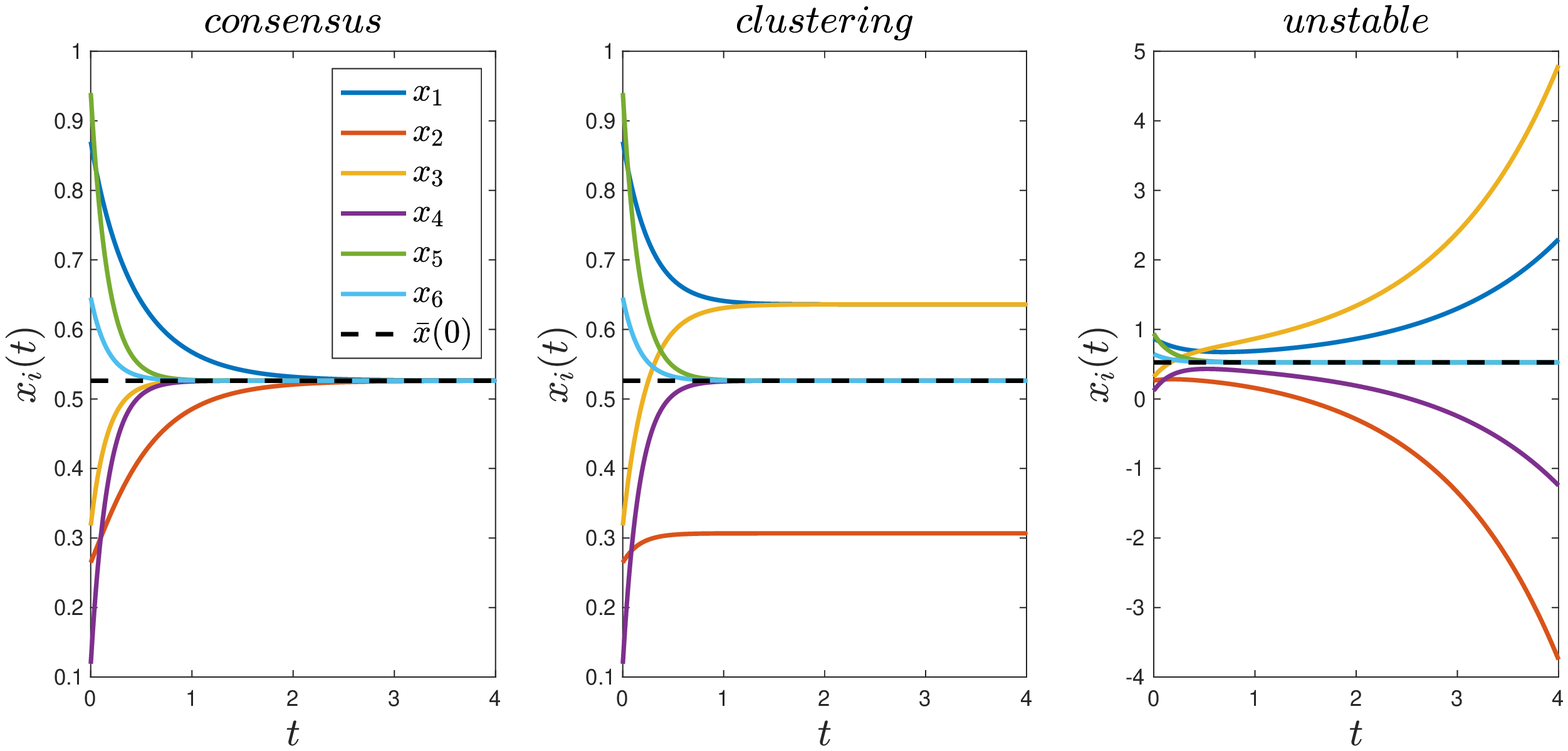}\\
\,\,\begin{tikzpicture}[scale=1]
	
    \node (a) at (1,-0.5) [] {$(a)$};	

	\node (n1) at (0.5,2) [circle,draw] {{\tiny 1}};
    \node (n2) at (1.5,2) [circle,draw] {{\tiny 2}};
    \node (n3) at (2,1) [circle,draw] {{\tiny 3}};
	\node (n4) at (1.5,0) [circle,draw] {{\tiny 4}};
    \node (n5) at (0.5,0) [circle,draw] {{\tiny 5}};
    \node (n6) at (0,1) [circle,draw] {{\tiny 6}};

		\draw [-,thick,dashed] (n1) -- (n2); 	
    	\draw [-,thick] (n1) -- (n3); 	
        \draw [-,thick] (n1) -- (n4);  	
	    \draw [-,thick] (n1) -- (n5); 	
	    \draw [-,thick] (n1) -- (n6);
	    \draw [-,thick] (n2) -- (n3);	
	    \draw [-,thick] (n2) -- (n4);  	
	    \draw [-,thick] (n2) -- (n5); 	
	    \draw [-,thick] (n2) -- (n6);
	    \draw [-,thick] (n3) -- (n4);	
	    \draw [-,thick] (n3) -- (n5);  	
	    \draw [-,thick] (n3) -- (n6);
	    \draw [-,thick] (n4) -- (n5);	
	    \draw [-,thick] (n4) -- (n6);  	
	    \draw [-,thick] (n5) -- (n6);	

\end{tikzpicture}
\,\,\,\begin{tikzpicture}[scale=1]

    \node (b) at (1,-0.5) [] {$(b)$};	
	
	\node (n1) at (0.5,2) [circle,draw] {{\tiny 1}};
    \node (n2) at (1.5,2) [circle,draw] {{\tiny 2}};
    \node (n3) at (2,1) [circle,draw] {{\tiny 3}};
	\node (n4) at (1.5,0) [circle,draw] {{\tiny 4}};
    \node (n5) at (0.5,0) [circle,draw] {{\tiny 5}};
    \node (n6) at (0,1) [circle,draw] {{\tiny 6}};

		\draw [-,thick,dashed] (n1) -- (n2); 	
    	\draw [-,thick] (n1) -- (n3); 	
        \draw [-,thick] (n1) -- (n4);  	
	    \draw [-,thick] (n1) -- (n5); 	
	    \draw [-,thick] (n1) -- (n6);
	    \draw [-,thick,dashed] (n2) -- (n3);	
	    \draw [-,thick] (n2) -- (n4);  	
	    \draw [-,thick] (n2) -- (n5); 	
	    \draw [-,thick] (n2) -- (n6);
	    \draw [-,thick] (n3) -- (n4);	
	    \draw [-,thick] (n3) -- (n5);  	
	    \draw [-,thick] (n3) -- (n6);
	    \draw [-,thick] (n4) -- (n5);	
	    \draw [-,thick] (n4) -- (n6);  	
	    \draw [-,thick] (n5) -- (n6);	

\end{tikzpicture}
\,\,\,\begin{tikzpicture}[scale=1]

    \node (c) at (1,-0.5) [] {$(c)$};	
	
	\node (n1) at (0.5,2) [circle,draw] {{\tiny 1}};
    \node (n2) at (1.5,2) [circle,draw] {{\tiny 2}};
    \node (n3) at (2,1) [circle,draw] {{\tiny 3}};
	\node (n4) at (1.5,0) [circle,draw] {{\tiny 4}};
    \node (n5) at (0.5,0) [circle,draw] {{\tiny 5}};
    \node (n6) at (0,1) [circle,draw] {{\tiny 6}};

		\draw [-,thick,dashed] (n1) -- (n2); 	
    	\draw [-,thick] (n1) -- (n3); 	
        \draw [-,thick] (n1) -- (n4);  	
	    \draw [-,thick] (n1) -- (n5); 	
	    \draw [-,thick] (n1) -- (n6);
	    \draw [-,thick,dashed] (n2) -- (n3);	
	    \draw [-,thick] (n2) -- (n4);  	
	    \draw [-,thick] (n2) -- (n5); 	
	    \draw [-,thick] (n2) -- (n6);
	    \draw [-,thick,dashed] (n3) -- (n4);	
	    \draw [-,thick] (n3) -- (n5);  	
	    \draw [-,thick] (n3) -- (n6);
	    \draw [-,thick] (n4) -- (n5);	
	    \draw [-,thick] (n4) -- (n6);  	
	    \draw [-,thick] (n5) -- (n6);	

\end{tikzpicture}
\caption{The collective behavior of the signed network (\ref{eq:consensus-protocol-overall})
on a 6-node complete network with different selections of negative
edges, namely, consensus, clustering and being unstable when the negative
edges are chosen to be $\left\{ (1,2)\right\} $ (a), $\left\{ (1,2),\thinspace(2,3)\right\} $
(b) and $\left\{ (1,2),\thinspace(2,3),\thinspace(3,4)\right\} $
(c), respectively.}
\label{fig:6-node-negative-coupling}
\end{figure}

In contrast to unsigned networks, typical collective behaviors of
signed network (\ref{eq:consensus-protocol-overall}) are more intricate,
include achieving (average) consensus, cluster consensus or even being
unstable (see Figure \ref{fig:6-node-negative-coupling} for illustrative
examples of complete signed networks). Especially, the case of being
unstable is not desirable for the most applications. 

A list of all types of interactions in terms of the sign of edge weight
is shown in Figure \ref{fig:type-of-interactions}(a) as well as the
illustration of their respective influence on the state evolution
of the associated agents in Figure \ref{fig:type-of-interactions}(b).
The state trajectories for a pair of neighboring agents $i$ and $j$
with different types of interactions are further provided in the right
panel in Figure \ref{fig:type-of-interactions}, stable (top) and
unstable (bottom) trajectories, respectively, where the state  evolution
are all initiated from $x_{i}(0)=-1$ and $x_{j}(0)=1$. Essentially,
a negative edge $(i,j)$ raises a positive feedback term in (\ref{eq:consensus-protocol-local})
which amplifies the difference between $x_{i}(t)$ and $x_{j}(t)$
(type \Rmnum{3} in Figure \ref{fig:type-of-interactions}(a)), while
a positive edge $(i,j)$ raises a negative feedback term in (\ref{eq:consensus-protocol-local})
which reduces the difference between $x_{i}(t)$ and $x_{j}(t)$ (type
\Rmnum{2} in Figure \ref{fig:type-of-interactions}(a)). Therefore,
a source of complexity in the collective behaviors of signed network
(\ref{eq:consensus-protocol-overall}) lies in the mix of both negative
and positive feedback loops in one network, which can exhibit diverse
collective behaviors. This observation motivates the recent work regarding
multi-agent systems on signed networks \cite{zelazo2014definiteness,zelazo2017robustness,chen2016definiteness,chen2016characterizing,pan2016laplacian,shi2019dynamics}. 

In the following discussions, we shall first provide a novel insight
into the stability of signed network (\ref{eq:consensus-protocol-overall})
and argue that for a signed network containing negative cut set, further
actions are necessary to stabilize the unstable signed network. Then,
we propose a distributed stabilization protocol for signed networks
via self-loop compensation which is essentially introducing a damping
term in the interaction protocol of those agents that are incident
to negative edges.

\begin{figure*}[t]
\begin{centering}
\begin{tikzpicture}[scale=0.85]


\node (n1) at (0,0) [circle,draw] {$i$};
\node (n2) at (2,0) [circle,draw] {{\small $j$}};

\node (type) at (-1,0) [] {Type \Rmnum{1}: \quad \quad};

	\path[]
	(n1) [->,thick,  bend left=12] edge node[midway, above] {} (n2); 

	\path[]
	(n2) [->,thick,  bend left=12] edge node[midway, below] {} (n1); 


 \node (n1) at (0+4,0) [circle,draw] {};
    \node (x1) at (-0.5+4,0) []  {$\dot{x}_i$};

    \node (n2) at (2+4,0) [circle,draw] {};
    \node (x2) at (2.5+4,0) []  {$\dot{x}_j$};

    \node (n1v0) at (0.8+4,0) [] {}; 
	\path[]
    (n1v0) [<-,thick] edge node[right] {} (n1);

    \node (n2v0) at (1.2+4,0) [] {}; 
	\path[]
    (n2v0) [<-,thick] edge node[right] {} (n2);



\node (n1) at (0,0-1) [circle,draw] {$i$};
\node (n2) at (2,0-1) [circle,draw] {{\small $j$}};

\node (type) at (-1,0-1) [] {Type \Rmnum{2}: \quad \quad};

	\path[]
	(n2) [->,thick,  bend left=12] edge node[midway, below] {} (n1); 


    \node (n1) at (0+4,0-1) [circle,draw] {};
    \node (x1) at (-0.5+4,0-1) []  {$\dot{x}_i$};

    \node (n2) at (2+4,0-1) [circle,draw] {};
    \node (x2) at (2.5+4,0-1) []  {$\dot{x}_j$};

    \node (n1v0) at (0.8+4,0-1) [] {}; 
	\path[]
    (n1v0) [<-,thick] edge node[right] {} (n1);



\node (n1) at (0,-2) [circle,draw] {$i$};
\node (n2) at (2,-2) [circle,draw] {{\small $j$}};

\node (type) at (-1,-2) [] {Type \Rmnum{3}: \quad \quad};

	\path[]
	(n2) [->,thick,dashed,  bend left=12] edge node[midway, below] {} (n1); 


    \node (n1) at (0+4,-2) [circle,draw] {};
    \node (x1) at (0.5+4,-2) []  {$\dot{x}_i$};

    \node (n2) at (2+4,-2) [circle,draw] {};
    \node (x2) at (2.5+4,-2) []  {$\dot{x}_j$};

    \node (n1v0) at (-0.8+4,-2) [] {}; 
	\path[]
    (n1v0) [<-,thick,dashed] edge node[right] {} (n1);


\node (n1) at (0,0-3) [circle,draw] {$i$};
\node (n2) at (2,0-3) [circle,draw] {{\small $j$}};

\node (type) at (-1,0-3) [] {Type \Rmnum{4}: \quad \quad};

	\path[]
	(n1) [->,thick,dashed,  bend left=12] edge node[midway, above] {} (n2); 

	\path[]
	(n2) [->,thick,dashed,  bend left=12] edge node[midway, below] {} (n1); 


  \node (n1) at (0+4,0-3) [circle,draw] {};
    \node (x1) at (0.5+4,0-3) []  {$\dot{x}_i$};

    \node (n2) at (2+4,0-3) [circle,draw] {};
    \node (x2) at (1.5+4,0-3) []  {$\dot{x}_j$};

    \node (n1v0) at (-0.8+4,0-3) [] {}; 
	\path[]
    (n1v0) [<-,thick,dashed] edge node[right] {} (n1);

    \node (n2v0) at (2.8+4,0-3) [] {}; 
	\path[]
    (n2v0) [<-,thick,dashed] edge node[right] {} (n2);


\node (n1) at (0,0-4) [circle,draw] {$i$};
\node (n2) at (2,0-4) [circle,draw] {{\small $j$}};

\node (type) at (-1,0-4) [] {Type \Rmnum{5}: \quad \quad};

	\path[]
	(n1) [->,thick,  bend left=12] edge node[midway, above] {} (n2); 

	\path[]
	(n2) [->,thick,dashed,  bend left=12] edge node[midway, below] {} (n1); 


   \node (n1) at (0+4,0-4) [circle,draw] {};
    \node (x1) at (0.5+4,0-4) []  {$\dot{x}_i$};

    \node (n2) at (2+4,0-4) [circle,draw] {};
    \node (x2) at (2.5+4,0-4) []  {$\dot{x}_j$};

    \node (n1v0) at (-0.8+4,0-4) [] {}; 
	\path[]
    (n1v0) [<-,thick,dashed] edge node[right] {} (n1);

    \node (n2v0) at (1.2+4,0-4) [] {}; 
	\path[]
    (n2v0) [<-,thick] edge node[right] {} (n2);

\node (a) at (0.9,1-5.9) [] {$(a)$};

\node (b) at (5,1-5.9) [] {$(b)$};

\end{tikzpicture}\,\,\,\,\,\,\,\,\,\,\includegraphics[width=7cm]{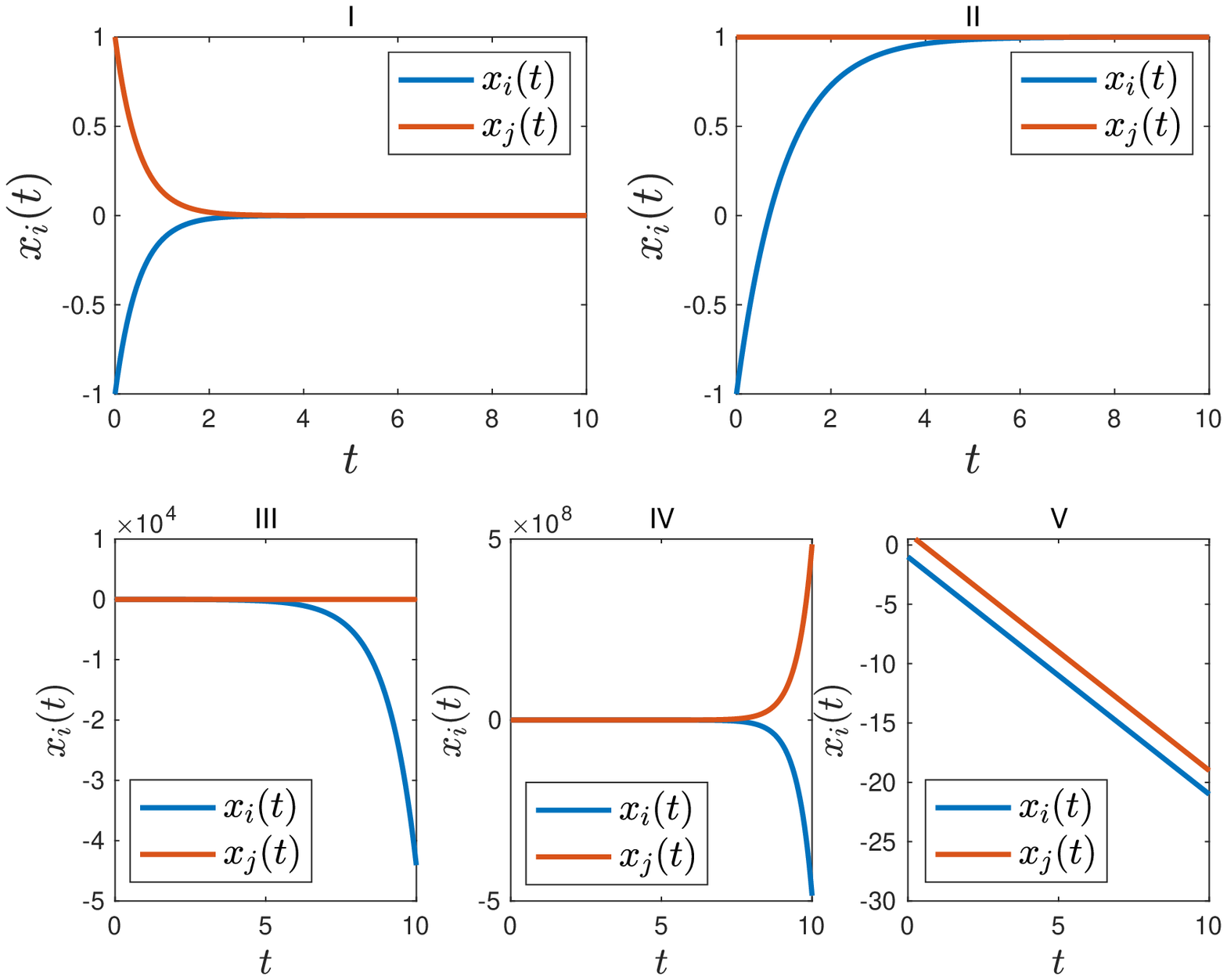}
\par\end{centering}
\caption{Types of pairwise interactions amongst a pair of neighboring agents
(a) and their respective influence on the state evolution of the associated
agents (b). Stable pairwise interactions in Figure \ref{fig:type-of-interactions},
namely, types \Rmnum{1} and \Rmnum{2} (top) and unstable pairwise
interactions in Figure \ref{fig:type-of-interactions}, namely, types
\Rmnum{3}, \Rmnum{4} and \Rmnum{5} (bottom). }

\label{fig:type-of-interactions}
\end{figure*}

\section{Negative Cut Set}

We first examine a graph-theoretic object that determines the stability
of the signed network (\ref{eq:consensus-protocol-overall}). In graph
theory, a cut edge (cut set) is an edge (a subset of edges) in $\mathcal{G}$
whose removal increases the number of connected component of $\mathcal{G}$.
Notably, cut edges in a graph are closely related to the transient
stability analysis of power systems \cite{song2017characterization}.
\begin{defn}
Let $\mathcal{G}=(\mathcal{V},\mathcal{E},W)$ be a signed network.
If all negative edges in $\mathcal{G}$ form a cut set $\mathcal{C}^{-}$,
then $\mathcal{C}^{-}$ is referred to as a \emph{negative cut set}. 
\end{defn}
\begin{figure}[tbh]
\begin{centering}
\begin{tikzpicture}[scale=0.9, >=stealth',   pos=.8,  photon/.style={decorate,decoration={snake,post length=1mm}} ]
	\node (n1) at (0,0) [circle,draw] {1};
	\node (n2) at (1.5,0) [circle,draw] {2};
    \node (n3) at (3,0) [circle,draw] {3};
    \node (n4) at (4.5,0) [circle,draw] {4};
	\node (n5) at (0,1.5) [circle,draw] {5};
	\node (n6) at (1.5,1.5) [circle,draw] {6};
    \node (n7) at (3,1.5) [circle,draw] {7};
    \node (n8) at (4.5,1.5) [circle,draw] {8};	

    \node (G0) at (-1,0.75) [] {$\mathcal{G}_0:$};	

	\path[]

	(n1) [<->,thick] edge node[below] {} (n5)
	(n2) [<->,thick] edge node[below] {} (n6)
	(n3) [<->,thick] edge node[below] {} (n8)
	(n6) [<->,thick] edge node[below] {} (n7)
	(n2) [<->,thick] edge node[below] {} (n7);

	\path[]

	(n2) [<->,thick,dashed] edge node[below] {} (n3)
	(n1) [<->,thick,dashed] edge node[below] {} (n6)
	(n3) [<->,thick,dashed] edge node[below] {} (n4);

\end{tikzpicture}
\par\end{centering}
\caption{A signed network $\mathcal{G}_{0}$ where solid lines and dotted lines
represent positive and negative edges, respectively.}
\label{fig:8-node-signed-network-G}
\end{figure}
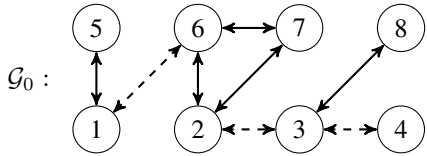

 Let $\mathcal{G}^{+}$ and $\mathcal{G}^{-}$ denote the positive
and negative components of $\mathcal{G}$, respectively, which are
 subgraphs of $\mathcal{G}$ with the same vertex set but with either
positive and negative edges, respectively. Let $W^{+}=(w_{ij}^{+})\in\mathbb{R}^{n\times n}$
and $W^{-}=(w_{ij}^{-})\in\mathbb{R}^{n\times n}$ be the adjacency
matrix of $\mathcal{G}^{+}$ and $\mathcal{G}^{-}$, respectively,
where $w_{ij}^{+}=\max\left\{ w_{ij},0\right\} $ and $w_{ij}^{-}=\min\left\{ w_{ij},0\right\} $.
Denote by $\tau(\mathcal{G}^{+})$ and $\tau(\mathcal{G}^{+})$ as
the number of connected component in $\mathcal{G}^{+}$ and $\mathcal{G}^{-}$,
respectively. 
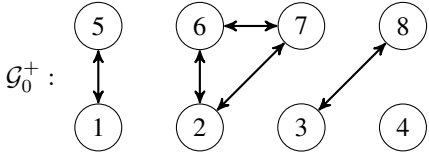
\begin{figure}[tbh]
\begin{centering}
\begin{tikzpicture}[scale=0.9, >=stealth',   pos=.8,  photon/.style={decorate,decoration={snake,post length=1mm}} ]
	\node (n1) at (0,0) [circle,draw] {1};
	\node (n2) at (1.5,0) [circle,draw] {2};
    \node (n3) at (3,0) [circle,draw] {3};
    \node (n4) at (4.5,0) [circle,draw] {4};
	\node (n5) at (0,1.5) [circle,draw] {5};
	\node (n6) at (1.5,1.5) [circle,draw] {6};
    \node (n7) at (3,1.5) [circle,draw] {7};
    \node (n8) at (4.5,1.5) [circle,draw] {8};	

   \node (G+) at (-1,0.75) [] {$\mathcal{G}_0^{+}:$};	

	\path[]

	(n1) [<->,thick] edge node[below] {} (n5)
	(n2) [<->,thick] edge node[below] {} (n6)
	(n3) [<->,thick] edge node[below] {} (n8)
	(n6) [<->,thick] edge node[below] {} (n7)
	(n2) [<->,thick] edge node[below] {} (n7);

\end{tikzpicture}
\par\end{centering}
\caption{The positive component $\mathcal{G}_{0}^{+}$ of $\mathcal{G}_{0}$
in Figure \ref{fig:8-node-signed-network-G}.}
\label{fig:8-node-signed-network-G-plus}
\end{figure}
\begin{figure}[tbh]
\begin{centering}
\begin{tikzpicture}[scale=0.9, >=stealth',   pos=.8,  photon/.style={decorate,decoration={snake,post length=1mm}} ]
	\node (n1) at (0,0) [circle,draw] {1};
	\node (n2) at (1.5,0) [circle,draw] {2};
    \node (n3) at (3,0) [circle,draw] {3};
    \node (n4) at (4.5,0) [circle,draw] {4};
	\node (n5) at (0,1.5) [circle,draw] {5};
	\node (n6) at (1.5,1.5) [circle,draw] {6};
    \node (n7) at (3,1.5) [circle,draw] {7};
    \node (n8) at (4.5,1.5) [circle,draw] {8};	

  \node (G-) at (-1,0.75) [] {$\mathcal{G}_0^{-}:$};	

	\path[]

	(n2) [<->,thick,dashed] edge node[below] {} (n3)
	(n1) [<->,thick,dashed] edge node[below] {} (n6)
	(n3) [<->,thick,dashed] edge node[below] {} (n4);

\end{tikzpicture}
\par\end{centering}
\caption{The negative component $\mathcal{G}_{0}^{-}$ of $\mathcal{G}_{0}$
in Figure \ref{fig:8-node-signed-network-G}.}
\label{fig:8-node-signed-network-G-negative}
\end{figure}
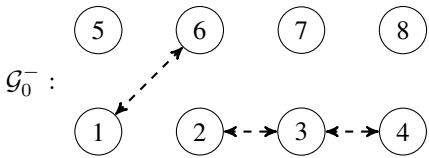

\begin{example}
For instance, consider the $\mathcal{G}_{0}$ in Figure \ref{fig:8-node-signed-network-G}
with positive and negative components $\mathcal{G}_{0}^{+}$ and $\mathcal{G}_{0}^{-}$
shown in Figures \ref{fig:8-node-signed-network-G-plus} and \ref{fig:8-node-signed-network-G-negative},
respectively. One can see that $\tau(\mathcal{G}_{0}^{+})=4$ and
$\tau(\mathcal{G}_{0}^{-})=5$ in this example. 
\end{example}
Let $i_{-}(\mathcal{L})$ denote the number negative eigenvalues of
signed Laplacian $\mathcal{L}(\mathcal{G})$ in (\ref{eq:consensus-protocol-overall}). 
\begin{prop}
The signed network (\ref{eq:consensus-protocol-overall}) is stable
if and only if $i_{-}(\mathcal{L})=0.$
\end{prop}
We have the estimation of number of negative eigenvalues of signed
Laplacian in terms of number of positive and negative components in
the following result. 
\begin{lem}
\cite{bronski2014spectral} \label{lem:bronski2014spectral} Let $\mathcal{G}$
be a connected signed network. Then $\tau(\mathcal{G}^{+})-1\leq i_{-}(\mathcal{L})\leq n-\tau(\mathcal{G}^{-})$.
\end{lem}
\begin{thm}
\label{thm:negative-cut-set} Let $\mathcal{G}=(\mathcal{V},\mathcal{E},W)$
be a connected signed network. 

1) If $\mathcal{G}^{+}$ is not connected, then the signed network
(\ref{eq:consensus-protocol-overall}) is not stable. 

2) If $\mathcal{G}$ has a negative cut set $\mathcal{C}^{-}$ whose
elements are all cut edges, then $i_{-}(\mathcal{L})=|\mathcal{C}^{-}|$.

3) If $\mathcal{G}$ has a negative cut set $\mathcal{C}^{-}$ whose
elements are not all cut edges, then $i_{-}(\mathcal{L})\leq|\mathcal{C}^{-}|$.
\end{thm}
\begin{proof}
1) Note that if $\mathcal{G}^{+}$ is not connected, then $\tau(\mathcal{G}^{+})\ge2$.
Applying Lemma \ref{lem:bronski2014spectral}, we have $i_{-}(\mathcal{L})\ge\tau(\mathcal{G}^{+})-1\ge1$,
implying that the signed network (\ref{eq:consensus-protocol-overall})
is not stable.

2) We first examine the positive component $\mathcal{G}^{+}$. Note
that there are $|\mathcal{C}^{-}|$ negative cut edges in $\mathcal{C}^{-}$.
Then, the positive component $\mathcal{G}^{+}$ can be obtained from
removing edges in $\mathcal{C}^{-}$ one by one, and each removal
increases the number of connected component by one. Thus, $\tau(\mathcal{G}^{+})=|\mathcal{C}^{-}|+1.$

We proceed to examine the negative component $\mathcal{G}^{-}$. Let
$p$ denote the number of connected components with at least one negative
edge in $\mathcal{G}^{-}$ and each of these connected components
$c_{i}^{-}$ contains $\tau_{i}>0$ negative edges where $i\in\underline{p}$.
Note that all edges in $c_{i}^{-}$ are negative and each connected
component $c_{i}^{-}$ is a tree. Then the number of nodes contained
in $c_{i}^{-}$ is $\tau_{i}+1$. Let $\mathcal{V}\left(\left\{ c_{i}^{-}\right\} _{i\in\underline{p}}\right)$
denote the number of nodes contained in $\left\{ c_{i}^{-}\right\} _{i\in\underline{p}}$,
then the number of remaining isolated nodes (those nodes that are
not incident to a negative edge in $\mathcal{G}$ and thus each of
then is a connected component in $\mathcal{G}^{-}$) in $\mathcal{G}^{-}$
is 
\[
n-\mathcal{V}\left(\left\{ c_{i}^{-}\right\} _{i\in\underline{p}}\right)=n-{\displaystyle \sum_{i=1}^{p}(\tau_{i}+1)}.
\]
Then, the number of negative components of $\mathcal{G}$ can be computed
by
\begin{align*}
\tau(\mathcal{G}^{-}) & =n-\mathcal{V}\left(\left\{ c_{i}^{-}\right\} _{i\in\underline{p}}\right)+p\\
 & =n-{\displaystyle \sum_{i=1}^{p}(\tau_{i}+1)}+p=n-\sum_{i=1}^{p}\tau_{i}.
\end{align*}
We note that $\sum_{i=1}^{p}\tau_{i}=|\mathcal{G}^{-}|$ and then
by applying Lemma \ref{lem:bronski2014spectral}, we have 
\[
|\mathcal{C}^{-}|+1-1\leq i_{-}(\mathcal{L})\leq n-(n-|\mathcal{C}^{-}|),
\]
which implies $i_{-}(\mathcal{L})=|\mathcal{C}^{-}|$.

3) If $\mathcal{G}$ has a negative cut set $\mathcal{C}^{-}$ whose
elements are not all cut edges, then there are two cases. Case 1:
there are no cycles in $\mathcal{C}$, then according to the proof
of 2), $i_{-}(\mathcal{L})\leq|\mathcal{C}^{-}|$. Case 2: there exists
at least one cycle in $\mathcal{C}^{-}$, then according to the proof
of 2), one has $\mathcal{V}\left(\left\{ c_{i}^{-}\right\} _{i\in\underline{p}}\right)<\tau_{i}+1$,
therefore, $i_{-}(\mathcal{L})<|\mathcal{C}^{-}|$.
\end{proof}
\begin{rem}
According to Theorem \ref{thm:negative-cut-set}, the negative cut
set turns out to be an incentive for the instability of the signed
network (\ref{eq:consensus-protocol-overall}). Specifically, Theorem
implies the lower bound and upper bound of $i_{-}(\mathcal{L})$ in
Lemma \ref{lem:bronski2014spectral} coincide when $\mathcal{G}$
has a negative cut set whose elements are all cut edges. Notable,
similar arguments have been obtained using linear matrix inequality
(LMI) techniques in \cite{zelazo2017robustness}. However Theorem
\ref{thm:negative-cut-set} provides a novel insight into the correlation
between negative cut set and stability of signed network (\ref{eq:consensus-protocol-overall})
from a graph-theoretic perspective.
\end{rem}
We provide an example to illustrate Theorem \ref{thm:negative-cut-set}.
\begin{example}
Continue to consider the signed network $\mathcal{G}_{0}$ in Figure
\ref{fig:8-node-signed-network-G}. Compute the spectrum of signed
Laplacian in this example yields 
\[
\boldsymbol{\lambda}(\mathcal{L}(\mathcal{G}_{0}))=\left\{ -2.39,-1.27,-0.53,0,1.32,1.59,2.47,2.82\right\} .
\]
On the other hand, $\mathcal{G}_{0}$ contains a negative cut set
as follows,
\[
\mathcal{C}^{-}=\left\{ (1,6),(2,3),(3,4)\right\} ,
\]
and by Theorem \ref{thm:negative-cut-set}, one has 
\[
i_{-}(\mathcal{L}(\mathcal{G}_{0}))=\tau(\mathcal{G}_{0}^{+})-1=8-\tau(\mathcal{G}_{0}^{-})=3.
\]
\end{example}
\begin{cor}[Tree networks]
\label{cor:tree-unstable} Let $\mathcal{T}$ be a signed tree network.
Then, $i_{-}(\mathcal{L}(\mathcal{T}))$ is equal to the number of
negative edges in $\mathcal{T}$.
\end{cor}
\begin{cor}[Structurally balanced networks]
\label{cor:SB-unstable}Let $\mathcal{G}$ be a structurally balanced
signed network. Then, $i_{-}(\mathcal{L}(\mathcal{G}))$ is equal
to the number of negative edges in $\mathcal{G}$.
\end{cor}
Notably, according to Corollary \ref{cor:tree-unstable} and Corollary
\ref{cor:SB-unstable}, the multi-agent systems (\ref{eq:consensus-protocol-overall})
on tree or structurally balanced signed networks are always unstable.

Moreover, we proceed to examine how the negative edges influence the
stability of multi-agent systems on circle networks which play an
important role in the performance of consensus networks \cite{zelazo2013performance}.
\begin{thm}[Circle networks]
 \label{thm:circle} Let $\mathcal{G}$ be a circle network with
negative edge set $\mathcal{E}^{-}$. Then, 
\end{thm}
1) If $|\mathcal{E}^{-}|=n$, then $i_{-}(\mathcal{L}(\mathcal{G}))=n-1$.

2) If $|\mathcal{E}^{-}|<n$, then $|\mathcal{E}^{-}|-1\leq i_{-}(\mathcal{L}(\mathcal{G}))\leq|\mathcal{E}^{-}|$
.
\begin{proof}
If $|\mathcal{E}^{-}|=n$, note that $\tau(\mathcal{G}_{-})=1$ and
$\tau(\mathcal{G}_{+})=n$, then according to Lemma \ref{lem:bronski2014spectral},
one has $i_{-}(\mathcal{L}(\mathcal{G}))=n-1$. 

If $|\mathcal{E}^{-}|<n$, if the $|\mathcal{E}^{-}|$ edges form
one connected component, then $\tau(\mathcal{G}_{-})=n-|\mathcal{E}^{-}|$
and $\tau(\mathcal{G}_{+})=|\mathcal{E}^{-}|$, then according to
Lemma \ref{lem:bronski2014spectral}, one has $|\mathcal{E}^{-}|-1\leq i_{-}(\mathcal{L}(\mathcal{G}))\leq|\mathcal{E}^{-}|$. 

If $|\mathcal{E}^{-}|$ edges form $|\mathcal{E}^{-}|$ connected
components, then $\tau(\mathcal{G}_{-})=n-|\mathcal{E}^{-}|$ and
$\tau(\mathcal{G}_{+})=|\mathcal{E}^{-}|$, then according to Lemma
\ref{lem:bronski2014spectral}, one has $|\mathcal{E}^{-}|-1\leq i_{-}(\mathcal{L}(\mathcal{G}))\leq|\mathcal{E}^{-}|$.
If there are $m$ connected components with $n_{1},n_{2},\cdots,n_{m}$
edges in $\mathcal{E}^{-}$, respectively, and $n-{\displaystyle \sum_{i=1}^{m}}n_{i}$
disconnected components, then 

\[
\tau(\mathcal{G}_{-})=m+\left(n-\sum_{i=1}^{m}(n_{i}+1)\right)=n-|\mathcal{E}^{-}|,
\]
and 
\[
\tau(\mathcal{G}_{+})=\sum_{i=1}^{m}(n_{i}-1)+m-1=|\mathcal{E}^{-}|,
\]
then according to Lemma \ref{lem:bronski2014spectral}, one has $|\mathcal{E}^{-}|-1\leq i_{-}(\mathcal{L}(\mathcal{G}))\leq|\mathcal{E}^{-}|$. 
\end{proof}
According to the aforementioned discussions, if there exists a negative
cut set in a signed network, then the associated signed Laplacian
has at least one negative eigenvalue regardless of the magnitude of
negative weights, rendering the signed network (\ref{eq:consensus-protocol-overall})
unstable \cite{zelazo2017robustness,pan2016laplacian,zelazo2014definiteness}.\textcolor{magenta}{{}
}Moreover, Theorem \ref{thm:circle} also implies that a circle network
with more than two negative edges is not stable. Under this circumstance,
one may need to regain the stability of the unstable signed network
via local manipulation for its functionality \cite{song2017characterization,song2017network,song2019extension}.

Here, a reasonable and intuitive approach is to rebuild the diagonal
dominance of the signed Laplacian, more preferably, in a distributed
manner. From a local perspective, it turns out that the negative cut
set of a network plays a critical role in determining the number of
negative eigenvalues of signed Laplacian; from a global perspective,
the diagonal dominance of the signed Laplacian plays a central role
in determining the stability of signed network (\ref{eq:consensus-protocol-overall}),
which is often not free due to the existence of negative edges.

\section{Self-loop Compensation }

Upon the aforementioned analysis, retrieve diagonal dominance of signed
Laplacian via local-level adaptation is intuitively necessary to maintain
the stability of signed networks. To this end, an intuitive and straightforward
attempt is to compensate the diagonal entries of signed Laplacian
that can be fullfilled by introducing a damping term $c_{i}(t)=-k_{i}x_{i}(t)$
on top of (\ref{eq:consensus-protocol-local}),
\begin{align}
\dot{x}_{i}(t) & =\sum_{j\in\mathcal{N}_{i}}w_{ij}(x_{j}(t)-x_{i}(t))+c_{i}(t),i\in\mathcal{V},\label{eq:compensation-protocol}
\end{align}
where $k_{i}\ge0$ is the $i$-th entry in the compensation vector
$\boldsymbol{k}=[k_{1},k_{2},\cdots,k_{n}]^{T}\in\mathbb{R}^{n}$.
In view of this, the dynamics of signed network (\ref{eq:consensus-protocol-overall})
after self-loop compensation admits, 
\begin{equation}
\dot{\boldsymbol{x}}(t)=-\mathcal{L}^{\boldsymbol{k}}\boldsymbol{x}(t),\label{eq:feedback-overall}
\end{equation}
where $\mathcal{L}^{\boldsymbol{k}}=\mathcal{L}+\text{{\bf diag}}\{\boldsymbol{k}\}$.
Note that $\text{{\bf diag}}\{\boldsymbol{k}\}\succeq0$ is a positive
semidefinite diagonal matrix and $k_{i}=0$ implies agent $i$ does
not emply self-loop compensation. 

As we shall show that the self-loop compensation is triggered only
for agents who are incident to negative edges, denoted by, 
\begin{equation}
\mathcal{V}^{-}=\left\{ i\in\mathcal{V}\thinspace|\thinspace\exists j\in\mathcal{N}_{i}\thinspace\text{such that}\thinspace w_{ij}<0\right\} .
\end{equation}
Moreover, as we shall show in the upcoming discussions, all agents
that are incident to negative edges have to be compensated, otherwise,
the network may not be stable. 
\begin{rem}
The self-loop compensation is plausible since it can be applied to
guarantee the network stability and  even consensus in a fully distributed
fashion. For instance, if one conservatively chooses $\boldsymbol{k}=\boldsymbol{\delta}=[\delta_{1},\delta_{2},\cdots,\delta_{n}]^{T}\in\mathbb{R}^{n}$
where 
\begin{align}
\delta_{i} & =\sum_{j=1}^{n}\left(|w_{ij}|-w_{ij}\right),i\in\mathcal{V},\label{eq:compensation-vector}
\end{align}
The multi-agent system (\ref{eq:feedback-overall}) can achieve the
bipartite consensus or trivial consensus towards origin depending
on whether the underlying signed network is structurally balanced
\cite{altafini2013consensus}. Clearly, $\delta_{i}>0$ for all $i\in\mathcal{V}^{-}$
and $\delta_{i}=0$ for all $i\not\in\mathcal{V}^{-}$, and it is
notable that all necessary information to construct $\delta_{i}$
for each agent $i$ is locally accessible. In this case, one can employ
the elegant property of signed Laplacian (being positive semidefiniteness
with one zero eigenvalue) for consensus \cite{chen2020spectral,song2019extension,chen2016semidefiniteness,altafini2013consensus}.
Interestingly, $\boldsymbol{k}=\boldsymbol{\delta}$ also acts as
a critical boundary that renders $\mathcal{L}^{\boldsymbol{k}}$ (weak)
diagonally dominant \cite{horn1990matrix}.
\end{rem}

\begin{rem}
In the setting of (\ref{eq:compensation-protocol}), one can view
the magnitude of self-loop compensation $k_{i}$ as the weight $w_{ii}$
associated with self-loop edge in the related loopy graph \cite{song2017local}.
Notable physical interpretations of self-loop in networks can be conductance,
loads or  dissipation in the context of electrical networks \cite{dorfler2012kron,song2017local}.
Another  interpretation of self-loop compensation can be exerting
a virtual leader on these agents that are incident to negative edges,
steering whom towards origin (see Figure \ref{fig:self-compensation}).
The self-loop compensation term signifies the tendency of $x_{i}(t)$
that evolves towards the origin ($x_{0}=0$) \cite{cao2012distributed}.
However, different from the unsigned networks, whether the compensated
signed network (\ref{eq:feedback-overall}) can track the virtual
leader (origin) depends on the structure of the signed network as
well as the compensation vector.
\begin{figure}[h]
\begin{centering}
\begin{tikzpicture}[scale=1]
	
    \node (n2) at (1,1) [circle,draw] {$i$};

	\path[]

	(n2) edge [thick,loop above=15] node {$w_{ii}=k_i>0$} (n2);

\end{tikzpicture}\,\,\,\,\,\,\,\,\,\,\,\,\,\,\,\,\,\,\,\,\,\,\,\,\,\,\,\,\begin{tikzpicture}[scale=1]
	
    \node (n2) at (1.5,1.5) [circle,draw] {};
    \node (x2) at (1.5,2) []  {$x_i$};

    \node (n0) at (3.5,1.5) [circle,draw] {};
    \node (x0) at (3.6,2) []  {$x_0=0$};

    \node (n2v0) at (2.25,1.5) [] {}; 
	\path[]
    (n2v0) [<-,thick] edge node[right] {} (n2);


\end{tikzpicture}
\par\end{centering}
\caption{The self-loop compensation for an agent $i\in\mathcal{V}$.}

\label{fig:self-compensation}
\end{figure}
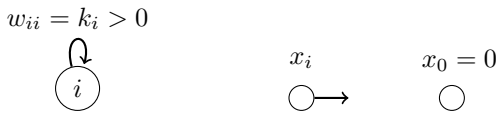
\end{rem}
However, it is worth noting that one cannot compensate the diagonal
entries of a signed Laplacian excessively, since the damping mechanism
can be time- and resource-consuming and may also lead to a loss of
the behavioral diversity of the compensated network--only has the
trivial consensus towards origin or even trigger unstability of the
network (see Figure (\ref{fig:cs-27})). We provide an illustrative
example to show this point as follows.

\begin{figure}[tbh]
\begin{centering}
\begin{tikzpicture}[scale=0.9, >=stealth',   pos=.8,  photon/.style={decorate,decoration={snake,post length=1mm}} ]
	\node (n1) at (0,0) [circle,draw] {1};
	\node (n2) at (1.5,0) [circle,draw] {2};
    \node (n3) at (3,0) [circle,draw] {3};
    \node (n4) at (4.5,0) [circle,draw] {4};
	\node (n5) at (0,1.5) [circle,draw] {5};
	\node (n6) at (1.5,1.5) [circle,draw] {6};
    \node (n7) at (3,1.5) [circle,draw] {7};
    \node (n8) at (4.5,1.5) [circle,draw] {8};	

    \node (G1) at (-1,0.75) [] {$\mathcal{G}_1:$};	

	\path[]
	(n5) [->,thick, bend right=8] edge node[below] {} (n6);
	\path[]
	(n6) [->,thick, bend right=8] edge node[below] {} (n5);

	\path[]
	(n5) [->,thick, bend right=8] edge node[below] {} (n1);
	\path[]
	(n1) [->,thick, bend right=8] edge node[below] {} (n5);

	\path[]
	(n6) [->,thick, bend right=8] edge node[below] {} (n1);
	\path[]
	(n1) [->,thick, bend right=8] edge node[below] {} (n6);

	\path[]
	(n6) [->,thick, bend right=8] edge node[below] {} (n2);
	\path[]
	(n2) [->,thick, bend right=8] edge node[below] {} (n6);

	\path[]
	(n6) [->,thick,dashed, bend right=8] edge node[below] {} (n7);
	\path[]
	(n7) [->,thick,dashed, bend right=8] edge node[below] {} (n6);

	\path[]
	(n2) [->,thick,dashed, bend right=8] edge node[below] {} (n7);
	\path[]
	(n7) [->,thick,dashed, bend right=8] edge node[below] {} (n2);

	\path[]
	(n3) [->,thick, dashed, bend right=8] edge node[below] {} (n2);
	\path[]
	(n2) [->,thick, dashed, bend right=8] edge node[below] {} (n3);

	\path[]
	(n3) [->,thick, bend right=8] edge node[below] {} (n7);
	\path[]
	(n7) [->,thick, bend right=8] edge node[below] {} (n3);

	\path[]
	(n7) [->,thick, bend right=8] edge node[below] {} (n8);
	\path[]
	(n8) [->,thick, bend right=8] edge node[below] {} (n7);

	\path[]
	(n3) [->,thick, bend right=8] edge node[below] {} (n8);
	\path[]
	(n8) [->,thick, bend right=8] edge node[below] {} (n3);

	\path[]
	(n3) [->,thick, bend right=8] edge node[below] {} (n4);
	\path[]
	(n4) [->,thick, bend right=8] edge node[below] {} (n3);

\end{tikzpicture}\\
\vskip 0.5cm
\begin{tikzpicture}[scale=0.9, >=stealth',   pos=.8,  photon/.style={decorate,decoration={snake,post length=1mm}} ]
	\node (n1) at (0,0) [circle,draw] {1};
	\node (n2) at (1.5,0) [circle,draw] {2};
    \node (n3) at (3,0) [circle,draw] {3};
    \node (n4) at (4.5,0) [circle,draw] {4};
	\node (n5) at (0,1.5) [circle,draw] {5};
	\node (n6) at (1.5,1.5) [circle,draw] {6};
    \node (n7) at (3,1.5) [circle,draw] {7};
    \node (n8) at (4.5,1.5) [circle,draw] {8};	

    \node (G2) at (-1,0.75) [] {$\mathcal{G}_2:$};

	\path[]
	(n5) [->,thick, bend right=8] edge node[below] {} (n6);
	\path[]
	(n6) [->,thick, bend right=8] edge node[below] {} (n5);

	\path[]
	(n5) [->,thick, bend right=8] edge node[below] {} (n1);
	\path[]
	(n1) [->,thick, bend right=8] edge node[below] {} (n5);

	\path[]
	(n6) [->,thick, bend right=8] edge node[below] {} (n1);
	\path[]
	(n1) [->,thick, bend right=8] edge node[below] {} (n6);

	\path[]
	(n6) [->,thick, bend right=8] edge node[below] {} (n2);
	\path[]
	(n2) [->,thick, bend right=8] edge node[below] {} (n6);

	\path[]
	(n6) [->,thick,dashed, bend right=8] edge node[below] {} (n7);
	\path[]
	(n7) [->,thick,dashed, bend right=8] edge node[below] {} (n6);

	\path[]
	(n3) [->,thick,dashed, bend right=8] edge node[below] {} (n2);
	\path[]
	(n2) [->,thick,dashed, bend right=8] edge node[below] {} (n3);

	\path[]
	(n3) [->,thick, bend right=8] edge node[below] {} (n7);
	\path[]
	(n7) [->,thick, bend right=8] edge node[below] {} (n3);

	\path[]
	(n7) [->,thick, bend right=8] edge node[below] {} (n8);
	\path[]
	(n8) [->,thick, bend right=8] edge node[below] {} (n7);

	\path[]
	(n3) [->,thick, bend right=8] edge node[below] {} (n8);
	\path[]
	(n8) [->,thick, bend right=8] edge node[below] {} (n3);

	\path[]
	(n3) [->,thick, bend right=8] edge node[below] {} (n4);
	\path[]
	(n4) [->,thick, bend right=8] edge node[below] {} (n3);

	\path[]
	(n2) [<-,thick, bend left=8] edge node[above] {} (n7);
	\path[]
	(n2) [->,thick, dashed, bend right=8] edge node[above] {} (n7);

\end{tikzpicture}
\par\end{centering}
\caption{A structurally balanced signed network $\mathcal{G}_{1}$ (top) and
a structurally imbalanced signed network $\mathcal{G}_{2}$ (bottom),
where solid lines and dotted lines represent positive edge and negative
edge, respectively.}
\label{fig:8-node-signed-network}
\end{figure}
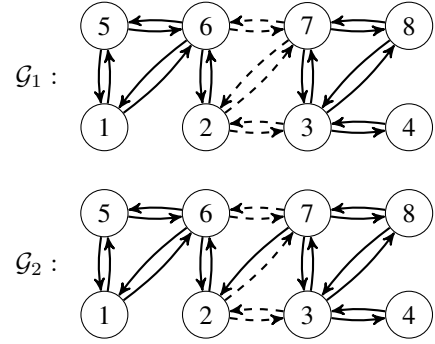

\begin{example}
\label{exa:CV-stability}In this example, we shall illustrate the
eigenvalue transition of signed Laplacian as a function of the magnitude
of compensation vector. Let the compensation vector be such that $\boldsymbol{k}(q)=q\boldsymbol{\delta}$
where $q\ge0$. Note that the scalar $q$ quantifies the magnitude
of the self-loop compensation $\boldsymbol{k}(q)$. Let $\mathcal{S}_{c}$
denote the set of agents employing self-loop compensation.

Case 1: $\mathcal{S}_{c}=\left\{ 2,3,6,7\right\} =\mathcal{V}^{-}$.
Remarkably, according to Figure \ref{fig:cs-2367}, that for structurally
balanced signed network $\mathcal{G}_{1}$ in Figure \ref{fig:8-node-signed-network},
the minimal magnitude of the compensation vector to render $-\mathcal{L}^{\boldsymbol{k}(q)}$
stable is realized exactly when $q=1$ (highlighted by the red vertical
lines in Figure \ref{fig:cs-2367}), in which case $\boldsymbol{k}(q)$
is equal to $\boldsymbol{\delta}$ in (\ref{eq:compensation-vector}).
However, for structurally imbalanced signed network $\mathcal{G}_{2}$
in Figure \ref{fig:8-node-signed-network}, the associated $-\mathcal{L}^{\boldsymbol{k}(q)}$
can be stable for some $q<1$ ($\boldsymbol{k}(q)<\boldsymbol{\delta}$).
Once can see the $-Re(\lambda_{1}(\mathcal{L}^{k}))$ is not a monotonic
function of $q$ for structurally imbalanced signed networks. This
implies the increment of magnitude of compensation vector does not
always enhance the stability of a compensated signed network. 
\begin{figure}[tbh]
\includegraphics[width=9cm]{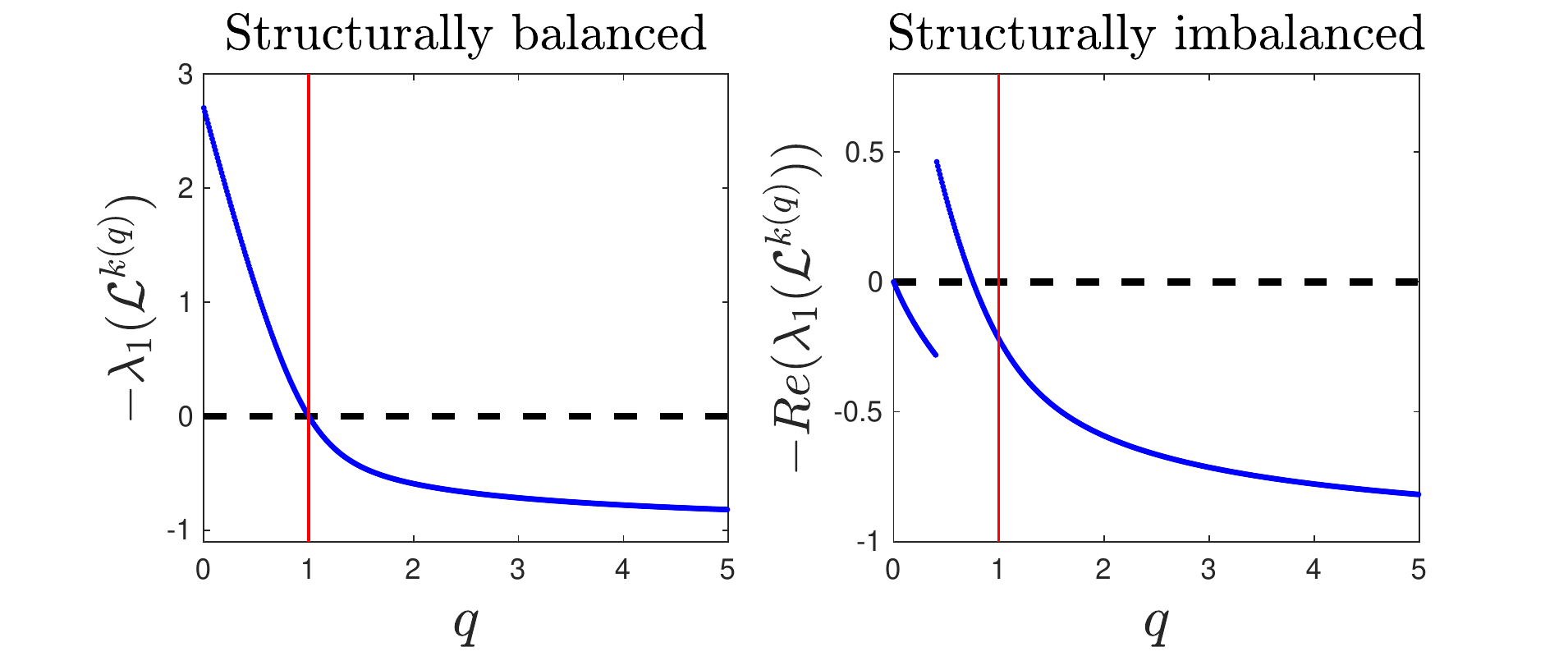}\caption{The smallest eigenvalue of $\mathcal{L}^{\boldsymbol{k}(q)}$ as a
function of $q$ for the structurally balanced network $\mathcal{G}_{1}$
and structurally imbalanced network $\mathcal{G}_{2}$ in Figure \ref{fig:8-node-signed-network},
respectively, where $\mathcal{S}_{c}=\left\{ 2,3,6,7\right\} $.}
\label{fig:cs-2367}
\end{figure}

Case 2: $\mathcal{S}_{c}=\left\{ 3,6,7\right\} \subset\mathcal{V}^{-}$.
\begin{figure}[tbh]
\includegraphics[width=9cm]{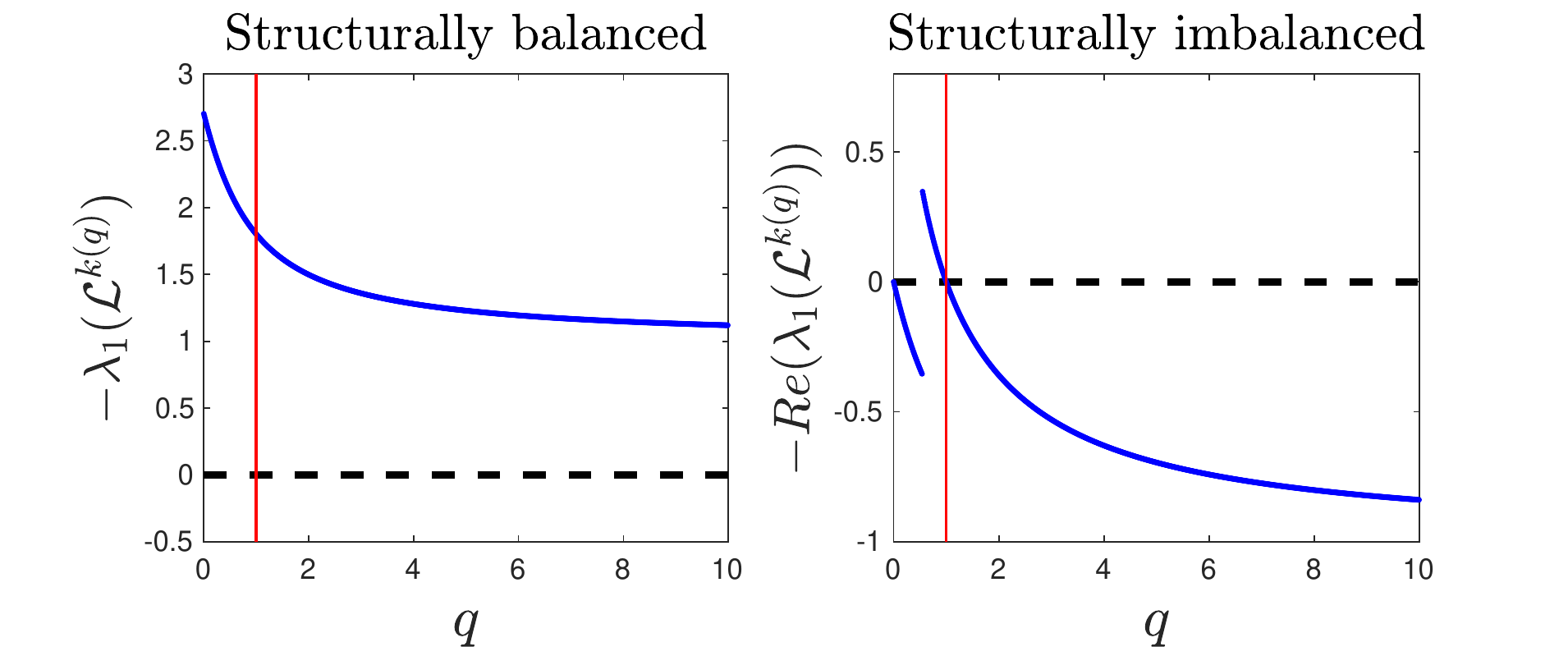}\caption{The smallest eigenvalue of $\mathcal{L}^{\boldsymbol{k}(q)}$ as a
function of $q$ for the structurally balanced network $\mathcal{G}_{1}$
and structurally imbalanced network $\mathcal{G}_{2}$ in Figure \ref{fig:8-node-signed-network},
respectively, where $\mathcal{S}_{c}=\left\{ 3,6,7\right\} $.}
\label{fig:cs-367}
\end{figure}
 In this case, only a subset of agents in $\mathcal{V}^{-}$ are compensated.
Note from Figure (\ref{fig:cs-367}) that the stability of compensated
signed network (\ref{eq:feedback-overall}) may exhibit polarization
phenomenon. Specifically, for structurally balanced signed networks,
the stability of (\ref{eq:feedback-overall}) cannot be enhanced by
the increment of $q$; however, stability of (\ref{eq:feedback-overall})
on structurally imbalanced signed networks can be guaranteed for a
large $q$ and also for some $q\in(0,1)$. 

Case 3: $\mathcal{S}_{c}=\left\{ 2,7\right\} \subset\mathcal{V}^{-}$.
In this case, fewer nodes in $\mathcal{V}^{-}$ are compensated, both
$-\lambda_{1}(\mathcal{L}^{\boldsymbol{k}(q)})$ and $-Re(\lambda_{1}(\mathcal{L}^{\boldsymbol{k}(q)}))$
goes to zero along with the increment of $q$ (top). Interestinglly,
when concentrating on the interval $q\in[0,2]$ for the case of structurally
imbalanced signed network, two intervals regarding the selection of
$q$ are observed which can render $-\mathcal{L}^{\boldsymbol{k}(q)}$
stable. 
\begin{figure}[tbh]
\includegraphics[width=9cm]{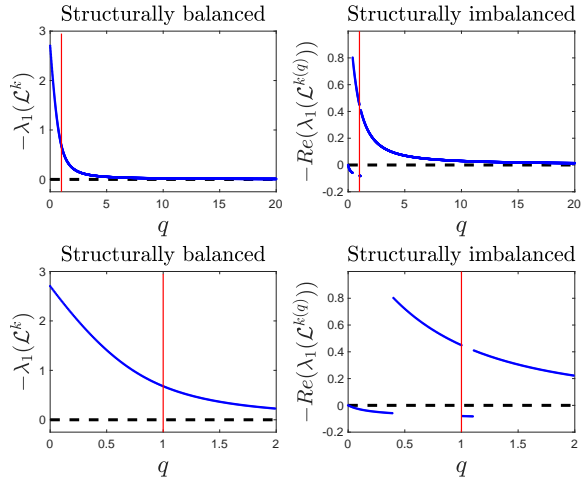}\caption{The smallest eigenvalue of $\mathcal{L}^{\boldsymbol{k}(q)}$ as a
function of $q$ for the structurally balanced network $\mathcal{G}_{1}$
and structurally imbalanced network $\mathcal{G}_{2}$ in Figure \ref{fig:8-node-signed-network},
respectively, where $\mathcal{S}_{c}=\left\{ 2,7\right\} $.}
\label{fig:cs-27}
\end{figure}
\end{example}

Inspired by the aforementioned discussions, it turns out that the
magnitude of the compensation vector $\boldsymbol{k}$ is closely
related to the structural balance of the underlying signed network.
Notably, the right panels in Figure \ref{fig:cs-2367}, Figure \ref{fig:cs-367}
and Figure \ref{fig:cs-27} also imply that the improper self-loop
compensation may also trigger unstability. Therefore, it is intricate
to examine an efficient selection of self-loop compensation for stabilizing
signed networks. Here, a natural question is how far away is a signed
network from being stable via self-loop compensation? In the following,
we shall examine the correlation between the compensation vector and
the stability of compensated network (\ref{eq:feedback-overall})
and explain the observations in Example \ref{exa:CV-stability} analytically.

\subsection{Undirected Signed Networks}

We first present fundamental facts related to structural balance of
a signed network. A Gauge transformation is performed by the matrix
$G=\text{{\bf diag}}\left\{ \sigma_{1},\sigma_{2},\cdots,\sigma_{n}\right\} :\mathbb{R}^{n\times n}\mapsto\mathbb{R}^{n\times n}$,
where $\sigma_{i}\in\{1,-1\}$ for all $i\in\underline{n}$\ \cite{altafini2013consensus}.
If a signed network $\mathcal{G}=(\mathcal{V},\mathcal{E},W)$ is
structurally balanced, then there exists a Gauge transformation $G$,
such that $GWG\ge0$. For each agent $i\in\mathcal{V}$ and an arbitrary
$x_{i}(0){\color{red}{\color{black}\in\mathbb{R}}}$, the multi-agent
system (\ref{eq:feedback-overall}) is said to admit bipartite consensus
if ${\color{black}{\color{blue}{\color{black}\lim}}}_{t\rightarrow\infty}|x_{i}(t)|=\alpha>0$.

We shall start our discussion from undirected signed networks. 

The stability of the Laplacian matrix of undirected signed networks
reduces to examining its positive semidefiniteness. Notably, for structurally
balanced, connected, undirected signed networks, Altafini has shown
that bipartite consensus can be achieved if $\boldsymbol{k}=\boldsymbol{\delta}$
\cite{altafini2013consensus}. Furthermore, we shall proceed to examine
the following three categories of compensation vectors.
\begin{thm}
\label{thm:SB-epsilon-gap-sufficient-necessary}Let $\mathcal{G}=(\mathcal{V},\mathcal{E},W)$
be a structurally balanced, connected, undirected signed network.
Then, the following statements hold.

1) If $\boldsymbol{k}\leq\boldsymbol{\delta}$ and $\boldsymbol{k}\not=\boldsymbol{\delta}$,
then the compensated network (\ref{eq:feedback-overall}) is unstable;

2) If $\boldsymbol{k}\ge\boldsymbol{\delta}$ and $\boldsymbol{k}\not=\boldsymbol{\delta}$,
then the compensated network (\ref{eq:feedback-overall}) achieves
trivial consensus;

3) If $\boldsymbol{k}\veebar\boldsymbol{\delta}$, then the compensated
network (\ref{eq:feedback-overall}) cannot achieve bipartite consensus.
\end{thm}
\begin{proof}
1) Assume that $\mathcal{L}^{\boldsymbol{k}}$ is positive semidefinite,
without  loss of generality, let $l_{11}^{\boldsymbol{k}}=\sum_{j=1}^{n}|w_{ij}|-\epsilon,$
where $\epsilon>0$, and $l_{ii}^{\boldsymbol{k}}=\sum_{j=1}^{n}|w_{ij}|$,
where $i\in\underline{n}$, and $i\neq1$. Thus $\mathcal{L}^{\boldsymbol{k}}=\mathcal{L}^{\boldsymbol{\delta}}-\Delta$,
where $\Delta=\text{{\bf diag}}\left\{ \epsilon,0,\ldots,0\right\} $.
Let $\boldsymbol{v}=[v_{1},v_{2},\ldots,v_{n}]^{T}\in\mathbb{R}^{n}$
be the eigenvector of $\mathcal{L}^{\boldsymbol{\delta}}$ corresponding
to the zero eigenvalue, note that $v_{i}\neq0$ for all $i\in\mathcal{V}$,
therefore $\boldsymbol{v}^{T}\mathcal{L}^{\boldsymbol{\delta}}\boldsymbol{v}=0$
and $\boldsymbol{v}^{T}\Delta\boldsymbol{v}>0.$ As a result, one
has,$\boldsymbol{v}^{T}\mathcal{L}^{\boldsymbol{\delta}}\boldsymbol{v}-\boldsymbol{v}^{T}\Delta\boldsymbol{v}<0$
which is a contradiction and $\mathcal{L}^{\boldsymbol{k}}$ has at
least one negative eigenvalue, i.e., the multi-agent system (\ref{eq:feedback-overall})
is unstable.

2) Note that $\mathcal{L}^{\boldsymbol{k}}=\mathcal{L}^{\boldsymbol{\delta}}+\text{{\bf diag}}\{\boldsymbol{k}-\boldsymbol{\delta}\}$,
then the proof follows from traditional treatments of leader-following
consensus problem via Gauge transformation \cite{cao2012distributed},
we shall omit it for space.

3) Assume that the multi-agent system (\ref{eq:feedback-overall})
achieves the bipartite consensus, then there exists a Gauge transformation
$G^{*}$ such that $\text{{\bf null}}(\mathcal{L}^{\boldsymbol{k}})=\text{{\bf span}}\{G^{*}\mathds{1}_{n}\}$. 

If $G^{*}=G$, then one has \textcolor{black}{$\underset{t\rightarrow\infty}{\text{{\bf lim}}}\boldsymbol{x}(t)=\frac{1}{n}G\mathds{1}_{n}\mathds{1}_{n}^{T}G\boldsymbol{x}(0)$
and $G\mathcal{L}^{\boldsymbol{k}}G\mathds{1}_{n}=\boldsymbol{0},$
which contradict with the fact $\boldsymbol{k}\neq\boldsymbol{\delta}$}.

If $G^{*}\neq G$, let $\mathcal{L}^{\boldsymbol{k}}=\mathcal{L}^{\boldsymbol{\delta}}+\varDelta$,
then one has 
\begin{align}
G^{*}\mathcal{L}^{\boldsymbol{k}}G^{*}\mathds{1}_{n} & =G^{*}(\mathcal{L}^{\boldsymbol{\delta}}+\varDelta)G^{*}\mathds{1}_{n}\nonumber \\
 & =G^{*}\mathcal{L}^{\boldsymbol{\delta}}G^{*}\mathds{1}_{n}+\varDelta\mathds{1}_{n},
\end{align}
due to $G^{*}\neq G$, thus $G^{*}\mathcal{L}^{\boldsymbol{\delta}}G^{*}\mathds{1}_{n}\geq\boldsymbol{0}$.
In addition, $\varDelta\mathds{1}_{n}$ has at least one element being
positive, Therefore, $G^{*}\mathcal{L}^{\boldsymbol{k}}G^{*}\mathds{1}_{n}\neq\boldsymbol{0}$,
i.e., $\text{{\bf null}}(\mathcal{L}^{\boldsymbol{k}})\neq\text{{\bf span}}\{G^{*}\mathds{1}_{n}\}$
and the multi-agent system (\ref{eq:feedback-overall}) cannot achieve
the bipartite consensus.
\end{proof}
According to Theorem \ref{thm:SB-epsilon-gap-sufficient-necessary},
one can obtain the following result regarding the bipartite consensus
of the compensated multi-agent network (\ref{eq:feedback-overall}).
\begin{cor}
\label{cor:1-1} Let $\mathcal{G}=(\mathcal{V},\mathcal{E},W)$ be
a structurally balanced, connected, undirected signed network. Then,
the compensated multi-agent network (\ref{eq:feedback-overall}) achieves
bipartite consensus if and only if $\boldsymbol{k}=\boldsymbol{\delta}$.
Moreover, the bipartite consensus value is 
\begin{equation}
\underset{t\rightarrow\infty}{\text{{\bf lim}}}\boldsymbol{x}(t)=\frac{1}{n}G\mathds{1}_{n}\mathds{1}_{n}^{T}G\boldsymbol{x}(0),
\end{equation}
where $G$ is the Gauge transformation associated with $\mathcal{G}$.
\end{cor}
\begin{rem}
One should differentiate Theorem \ref{thm:SB-epsilon-gap-sufficient-necessary}
from the related results on bipartitie consensus problem (e.g., Theorem
1 in \cite{altafini2013consensus}) in literatures, where necessary
and/or sufficient conditions are established for achieving bipartite
consensus for the following protocol, 
\begin{equation}
\dot{x}_{i}(t)=-\sum_{j\in\mathcal{N}_{i}}|w_{ij}|(x_{i}(t)-\text{{\bf sgn}}(w_{ij})x_{j}(t)),i\in\mathcal{V}.
\end{equation}
Essentially, Theorem \ref{thm:SB-epsilon-gap-sufficient-necessary}
indicates that the only case that compensated signed network (\ref{eq:feedback-overall})
admits bipartite consensus is $\boldsymbol{k}=\boldsymbol{\delta}$,
namely, other selections of compensation vector $\boldsymbol{k}\ge0$
will never lead the compensated signed network (\ref{eq:feedback-overall})
to a bipartite consensus solution. This is also valid for the selection
of compensation vector when $\mathcal{L}^{\boldsymbol{k}}$ is positive
semidefinite and has a simple zero eigenvalue, as stated below.
\end{rem}
\begin{cor}
\label{cor:1-2} Let $\mathcal{G}=(\mathcal{V},\mathcal{E},W)$ be
a structurally balanced, connected, undirected signed network. Then,
the $\mathcal{L}^{\boldsymbol{k}}\succeq0$ and has a simple zero
eigenvalue with eigenvector $G\mathbf{1}_{n}$ if and only if $\boldsymbol{k}=\boldsymbol{\delta}$,
where $G$ is the Gauge transformation associated with signed network
$\mathcal{G}$.
\end{cor}
In fact, for structurally balanced signed networks, the compensation
vector $\boldsymbol{k}=\boldsymbol{\delta}$ is optimal, as stated
in the following result.
\begin{thm}[Optimality]
\label{thm:SB-minma} Let $\mathcal{G}=(\mathcal{V},\mathcal{E},W)$
be a structurally balanced, connected, undirected signed network.
 Let $\boldsymbol{k}=[k_{1},\thinspace k_{2},\ldots,\thinspace k_{n}]^{T}\in\mathbb{R}^{n}$
be a compensation vector such that $\mathcal{L}^{\boldsymbol{k}}\succeq0$.
Then $\text{\ensuremath{\|}}\boldsymbol{k}\|_{1}\geq\|\boldsymbol{\delta}\|_{1}$.
\end{thm}
\begin{proof}
Suppose that $\text{\ensuremath{\|}}\boldsymbol{k}\|_{1}<\|\boldsymbol{\delta}\|_{1}$,
namely, 
\begin{equation}
\sum_{i=1}^{n}k_{i}<\sum_{i=1}^{n}\delta_{i}.
\end{equation}
Note that there exists a Gauge transformation $G$ such that $\mathds{1}_{n}^{T}G(\mathcal{L}+\text{{\bf diag}}\{\boldsymbol{\delta}\})G\mathds{1}_{n}\boldsymbol{=}0$.
Therefore,
\begin{align}
 & \mathds{1}_{n}^{T}G(\mathcal{L}+\text{{\bf diag}}\{\boldsymbol{k}\})G\mathds{1}_{n}\nonumber \\
 & =-\sum_{i=1}^{n}\sum_{j=1,j\neq i}^{n}|w_{ij}|\nonumber \\
 & +\sum_{i=1}^{n}\left(\sum_{j=1,j\neq i}^{n}w_{ij}+k_{i}\right)<0
\end{align}
which implies that $\mathcal{L}+\text{{\bf diag}}\{\boldsymbol{k}\}$
has a negative eigenvalue and, therefore, is not positive semidefinite.
This establishes a contradiction.
\end{proof}
In fact, Theorem \ref{thm:SB-minma} implies that the selection of
compensation vector stated in Theorem \ref{thm:SB-epsilon-gap-sufficient-necessary}
is optimal in term of the magnitude of the compensation vector, characterized
by $1$-norm of vectors. 

The above discussions are mainly concentrate on the structurally balanced
signed networks, and it is shown that the sufficient and necessary
condition for the multi-agent system (\ref{eq:feedback-overall})
achieving bipartite consensus is $\boldsymbol{k}=\boldsymbol{\delta}$;
for the structurally imbalanced signed network, it is shown that the
multi-agent system (\ref{eq:feedback-overall}) achieves trivial consensus
if $\boldsymbol{k}=\boldsymbol{\delta}$ \cite{altafini2013consensus}.
However, for a connected, structurally imbalanced signed network,
if the associated signed Laplacian matrix has negative eigenvalues,
a natural question is whether or not the compensation vector $\boldsymbol{k}$
has to be up to $\boldsymbol{\delta}$ so as to stabilize the system
(\ref{eq:feedback-overall})? We provide the following result to address
this question.
\begin{thm}
\label{prop:SB-2}Let $\mathcal{G}=(\mathcal{V},\mathcal{E},W)$ be
a structurally imbalanced, undirected signed network. Then, there
exists a compensation vector $\boldsymbol{k}^{\prime}=[k_{1}^{\prime},\thinspace k_{2}^{\prime},\ldots,\thinspace k_{n}^{\prime}]^{T}\in\mathbb{R}^{n}$
and $\boldsymbol{k}^{\prime}\le\boldsymbol{\delta}$ such that $-\mathcal{L}^{\boldsymbol{k}^{\prime}}$
is stable.
\end{thm}
\begin{proof}
As $\mathcal{G}$ is structurally imbalanced, the corresponding signed
Laplacian matrix $\mathcal{L}^{s}$ is positive definite with eigenvalues
ordered as $0<\lambda_{1}(\mathcal{L}^{s})\le\lambda_{2}(\mathcal{L}^{s})\le\cdots\le\lambda_{n}(\mathcal{L}^{s})$.
Chooses $\boldsymbol{k}^{\prime}$ such that $\boldsymbol{k}^{\prime}\le\boldsymbol{\delta}$
and 
\begin{equation}
|k_{1}^{\prime}-\delta_{1}|=\max_{i\in\underline{n}}|k_{i}^{\prime}-\delta_{i}|<\lambda_{1}(\mathcal{L}^{s}).
\end{equation}
Let

\begin{eqnarray}
\boldsymbol{\varDelta} & = & \boldsymbol{k}^{\prime}-\boldsymbol{\delta}\nonumber \\
 & = & [k_{1}^{\prime}-\delta_{1},k_{2}^{\prime}-\delta_{2},\cdots,k_{n}^{\prime}-\delta_{n}]^{T},
\end{eqnarray}
then we shall show that $\mathcal{L}^{s}+\text{{\bf diag}}\{\boldsymbol{\varDelta}\}$
is positive stable, which implies that $\mathcal{L}^{\boldsymbol{k}^{\prime}}$
is positive stable. To see this, suppose that there exists a  vector
$\boldsymbol{v}\in\mathbb{R}^{n}$ satisfying $\boldsymbol{v}^{T}\boldsymbol{v}=1$
and $\boldsymbol{v}^{T}(\mathcal{L}^{s}+\text{{\bf diag}}\{\boldsymbol{\varDelta}\})\boldsymbol{v}\leq0$;
then we have $\boldsymbol{v}^{T}\mathcal{L}^{s}\boldsymbol{v}\leq-\boldsymbol{v}^{T}\text{{\bf diag}}\{\boldsymbol{\varDelta}\}\boldsymbol{v}$.
Note that $\boldsymbol{v}^{T}\mathcal{L}^{s}\boldsymbol{v}\geq\lambda_{1}(\mathcal{L}^{s})$
and $-\boldsymbol{v}^{T}\text{{\bf diag}}\{\boldsymbol{\varDelta}\}\boldsymbol{v}\leq|k_{1}^{\prime}-\delta_{1}|$,
hence $\lambda_{1}(\mathcal{L}^{s})\leq|k_{1}^{\prime}-\delta_{1}|$
establishing a contradiction. 
\end{proof}
\begin{thm}
\label{prop:imbalanced-undirected} Let $\mathcal{G}=(\mathcal{V},\mathcal{E},W)$
be a structurally imbalanced, connected, undirected signed network.
Let $\lambda_{1}(\mathcal{L}^{\boldsymbol{\delta}})=\cdots=\lambda_{p}(\mathcal{L}^{\boldsymbol{\delta}})$
denote $1\le p\le n$ smallest eigenvalues of the matrix $\mathcal{L}^{\boldsymbol{\delta}}$
with corresponding normalized eigenvectors $\boldsymbol{v}_{1},\cdots,\boldsymbol{v}_{p}$.
Then, there exists a compensation vector $\boldsymbol{k}=\boldsymbol{\delta}-\lambda_{1}(\mathcal{L}^{\boldsymbol{\delta}})\mathds{1}_{n}$
such that $\mathcal{L}^{\boldsymbol{k}}$ is positive semidefinite
and has $p$ zero eigenvalues with corresponding eigenvectors $\boldsymbol{v}_{1},\cdots,\boldsymbol{v}_{p}$.
Moreover, the compensated multi-agent network (\ref{eq:feedback-overall})
achieves cluster consensus characterized by 
\begin{equation}
\underset{t\rightarrow\infty}{\text{{\bf lim}}}\boldsymbol{x}(t)={\displaystyle \sum_{j=1}^{p}}\boldsymbol{v}_{j}\boldsymbol{v}_{j}^{T}\boldsymbol{x}(0).
\end{equation}
\end{thm}
\begin{proof}
As $\mathcal{G}$ is structurally imbalanced, the corresponding signed
Laplacian matrix $\mathcal{L}^{\boldsymbol{\delta}}$ is positive
definite with eigenvalues ordered as 
\begin{equation}
0<\lambda_{1}(\mathcal{L}^{\boldsymbol{\delta}})\le\lambda_{2}(\mathcal{L}^{\boldsymbol{\delta}})\le\cdots\le\lambda_{n}(\mathcal{L}^{\boldsymbol{\delta}})
\end{equation}
and the corresponding eigenvectors $\boldsymbol{v}_{i}\in\mathbb{R}^{n}$,
where $i\in\underline{n}$. Due to,
\begin{equation}
\mathcal{L}+\text{{\bf diag}}\{\boldsymbol{k}\}=\mathcal{L}^{\boldsymbol{\delta}}-\text{{\bf diag}}\{\lambda_{1}(\mathcal{L}^{\boldsymbol{\delta}})\mathds{1}_{n}\},
\end{equation}
then for any $i\in\underline{n}$, one has,
\begin{equation}
\left(\mathcal{L}^{\boldsymbol{\delta}}-\text{{\bf diag}}\{\lambda_{1}(\mathcal{L}^{\boldsymbol{\delta}})\mathds{1}_{n}\}\right)\boldsymbol{v}_{i}=\left(\lambda_{i}(\mathcal{L}^{\boldsymbol{\delta}})-\lambda_{1}(\mathcal{L}^{\boldsymbol{\delta}})\right)\boldsymbol{v}_{i}.
\end{equation}
Therefore, the eigenvalues and its corresponding eigenvectors of the
matrix $\mathcal{L}+\text{{\bf diag}}\{\boldsymbol{k}\}$ are the
pairs $\left\{ \left(\lambda_{i}(\mathcal{L}^{\boldsymbol{\delta}})-\lambda_{1}(\mathcal{L}^{\boldsymbol{\delta}})\right),\,\boldsymbol{v}_{i}\right\} $.
Hence, $\mathcal{L}+\text{{\bf diag}}\{\boldsymbol{k}\}$ is positive
semidefinite and $\underset{t\rightarrow\infty}{\text{{\bf lim}}}\boldsymbol{x}(t)=\boldsymbol{v}_{1}\boldsymbol{v}_{1}^{T}\boldsymbol{x}(0)$.
\end{proof}
\begin{rem}
As indicated in Theorem \ref{prop:imbalanced-undirected}, the compensation
vector $\boldsymbol{k}$ for structurally imbalanced networks is not
necessary to be up to $\boldsymbol{\delta}$ to render the compensated
signed network (\ref{eq:feedback-overall}) stable, in which case,
$\mathcal{L}^{\boldsymbol{k}}$ is not necessary diagonal dominant.
In fact, different from the case of structurally balanced networks,
the signed Laplacian matrix for the structurally imbalanced networks
may have no negative eigenvalues. However, by employing the eigenvalue
of $\mathcal{L}^{\boldsymbol{\delta}}$ and vector $\boldsymbol{\delta}$,
Theorem \ref{prop:imbalanced-undirected} provides an elegant manner
to select the compensation vector $\boldsymbol{k}$ that can predict
the steady-state of the compensated signed network (\ref{eq:feedback-overall}).

One can see from Theorem \ref{prop:imbalanced-undirected} that the
multi-agent system (\ref{eq:feedback-overall}) can achieve the cluster
consensus when we choose $\boldsymbol{k}=\boldsymbol{\delta}-\lambda_{1}(\mathcal{L}^{\boldsymbol{\delta}})\mathds{1}_{n}$.
Similarly, a direct question here is related to the steady-state of
the multi-agent system (\ref{eq:feedback-overall}) in the case of
$\boldsymbol{k}>\boldsymbol{\delta}-\lambda_{1}(\mathcal{L}^{\boldsymbol{\delta}})\mathds{1}_{n}$
and $\boldsymbol{k}<\boldsymbol{\delta}-\lambda_{1}(\mathcal{L}^{\boldsymbol{\delta}})\mathds{1}_{n}$,
respectively.
\end{rem}
\begin{cor}
\label{cor:imbalanced-undirected-two}Let $\mathcal{G}=(\mathcal{V},\mathcal{E},W)$
be a structurally imbalanced, connected, undirected signed network.
Then,

Case 1: If the compensation vector $\boldsymbol{k}>\boldsymbol{\delta}-\lambda_{1}(\mathcal{L}^{\boldsymbol{\delta}})\mathds{1}_{n}$,
then $\mathcal{L}+\text{{\bf diag}}\{\boldsymbol{k}\}$ is positive
definite, i.e., the multi-agent system (\ref{eq:feedback-overall})
achieve the trivial consensus.

Case 2: If the compensation vector $\boldsymbol{k}<\boldsymbol{\delta}-\lambda_{1}(\mathcal{L}^{\boldsymbol{\delta}})\mathds{1}_{n}$,
then the multi-agent system (\ref{eq:feedback-overall}) is unstable.
\end{cor}
\begin{proof}
Case 1: Denote by 
\begin{equation}
\mathcal{L}+\text{{\bf diag}}\{\boldsymbol{k}\}=\mathcal{L}^{\boldsymbol{\delta}}-\text{{\bf diag}}\{\lambda_{1}(\mathcal{L}^{\boldsymbol{\delta}})\boldsymbol{1}_{n}\}+\varDelta
\end{equation}
where $\varDelta\in\mathbb{R}^{n\times n}$ is diagonal with $[\varDelta]_{ii}>0$
for all $i\in\underline{n}$. For any $\boldsymbol{\eta}\in\mathbb{R}^{n}$,
one has,
\begin{equation}
\boldsymbol{\eta}^{T}(\mathcal{L}^{\boldsymbol{\delta}}-\text{{\bf diag}}\{\lambda_{1}(\mathcal{L}^{\boldsymbol{\delta}})\mathds{1}_{n}\}+\varDelta)\boldsymbol{\eta}>0,
\end{equation}
therefore, $\mathcal{L}+\text{{\bf diag}}\{\boldsymbol{k}\}$ is positive
definite and the multi-agent system (\ref{eq:feedback-overall}) achieve
the trivial consensus.

Case 2: Denote by 
\begin{equation}
\mathcal{L}+\text{{\bf diag}}\{\boldsymbol{k}\}=\mathcal{L}^{\boldsymbol{\delta}}-\text{{\bf diag}}\{\lambda_{1}(\mathcal{L}^{\boldsymbol{\delta}})\mathds{1}_{n}\}+\varDelta,
\end{equation}
where $\varDelta\in\mathbb{R}^{n\times n}$ is diagonal and $[\varDelta]_{ii}<0$
for all $i\in\underline{n}$. Let $\boldsymbol{\eta}\in\mathbb{R}^{n}$
be such that $(\mathcal{L}^{\boldsymbol{\delta}}-\text{{\bf diag}}\{\lambda_{1}(\mathcal{L}^{\boldsymbol{\delta}})\boldsymbol{1}_{n}\})\boldsymbol{\eta}=\boldsymbol{0}$,
then one has,
\begin{equation}
\boldsymbol{\eta}^{T}(\mathcal{L}^{\boldsymbol{\delta}}-\text{{\bf diag}}\{\lambda_{1}(\mathcal{L}^{\boldsymbol{\delta}})\boldsymbol{1}_{n}\}+\varDelta)\boldsymbol{\eta}<0,
\end{equation}
therefore, there exists a negative eigenvalue for the matrix $\mathcal{L}+\text{{\bf diag}}\{\boldsymbol{k}\}$
and the multi-agent system (\ref{eq:feedback-overall}) is unstable.
\end{proof}

To sum up, there is no gap between $\boldsymbol{k}$ and $\boldsymbol{\delta}$
to render the multi-agent system (\ref{eq:feedback-overall}) achieving
the bipartite consensus if the underlying signed network is structurally
balanced according to Corollary \ref{cor:1-1}. Furthermore, Theorem
\ref{thm:SB-minma} indicates that the minimum total magnitude of
the self-loop compensation is $\boldsymbol{1}_{n}^{T}\boldsymbol{\delta}$
for structurally balanced signed networks, whereas, the magnitude
of compensation can be less for structurally imbalanced signed networks
(as it is shown in Example 1) and Theorem \ref{prop:imbalanced-undirected}
provides an approach to choose this compensation vector for predictable
behavior of the resultant compensated network.

\subsection{Directed Signed Networks}

We now proceed to examine directed signed networks. First recall the
following basic facts for directed networks. A signed network $\mathcal{G}=(\mathcal{V},\mathcal{E},W)$
is weight balanced if $\sum_{j=1}^{n}|w_{ij}|=\sum_{j=1}^{n}|w_{ji}|$
for all $i\in\mathcal{V}$. A matrix $M=[m_{ij}]\in\mathbb{R}^{n\times n}$
is irreducible if the indices $\left\{ 1,2,\cdots,n\right\} $ cannot
be decomposed into two disjoint non-empty subsets $\left\{ i_{1},i_{2},\cdots,i_{n_{1}}\right\} $
and $\left\{ j_{1},j_{2},\cdots,j_{n_{2}}\right\} $ where $n_{1}+n_{2}=n$
such that $m_{i_{\alpha}j_{\text{\ensuremath{\beta}}}}=0$ for all
$\alpha\in\left\{ 1,2,\cdots,n_{1}\right\} $ and $\beta\in\left\{ 1,2,\cdots,n_{2}\right\} $.
The graph of a matrix $M=[m_{ij}]\in\mathbb{R}^{n\times n}$, denoted
by $\mathcal{G}(M)$, is such that $(i,j)\in\mathcal{E}$ if and only
if $m_{ij}\not=0$ for all $i,j\in\mathcal{V}$. A matrix $M$ is
irreducible if and only if $\mathcal{G}(M)$ is strongly connected.

Similar to the undirected signed network, we shall first discuss the
condition that the compensation vector should satisfy to render the
multi-agent system (\ref{eq:feedback-overall}) stable.
\begin{thm}
\label{thm:SB-epsilon-gap-sufficient-necessary-directed}Let $\mathcal{G}=(\mathcal{V},\mathcal{E},W)$
be a structurally balanced, strongly connected, directed signed network.
Then, the following statements hold.

1) If $\boldsymbol{k}\leq\boldsymbol{\delta}$ and $\boldsymbol{k}\not=\boldsymbol{\delta}$,
then the compensated network (\ref{eq:feedback-overall}) is unstable;

2) If $\boldsymbol{k}\ge\boldsymbol{\delta}$ and $\boldsymbol{k}\not=\boldsymbol{\delta}$,
then the compensated network (\ref{eq:feedback-overall}) achieves
trivial consensus;

3) If $\boldsymbol{k}\veebar\boldsymbol{\delta}$, then the compensated
network (\ref{eq:feedback-overall}) cannot achieve bipartite consensus.
\end{thm}
\begin{proof}
1) Without loss of generality, let 
\begin{equation}
l_{11}^{\boldsymbol{k}}=\sum_{j=1}^{n}|w_{1j}|-\epsilon,
\end{equation}
where $\epsilon>0$, and $l_{ii}^{\boldsymbol{k}}=\sum_{j=1}^{n}|w_{ij}|$,
where $i\in\underline{n}$, and $i\neq1$. Then $\mathcal{L}^{\boldsymbol{k}}=\mathcal{L}^{\delta}-\Delta$,
where $\Delta=\text{{\bf diag}}\left\{ \epsilon,0,\ldots,0\right\} .$ 

Note that $\mathcal{G}$ is structurally balanced, thus there exists
a Gauge transformation $G$ satisfying $\mathcal{L}^{\delta}=G\mathcal{L}^{\prime}G$,
where $\mathcal{L}^{\prime}=[l_{ij}^{\prime}]\in\mathbb{R}^{n\times n}$
such that $l_{ij}^{\prime}=-|l_{ij}^{\delta}|$ for all $i\not=j\in\underline{n}$
and $l_{ii}^{\prime}=l_{ii}^{\delta}$ for all $i\in\underline{n}$.
Hence 
\begin{align}
\mathcal{L}^{\boldsymbol{k}} & =\mathcal{L}^{\delta}-\Delta\nonumber \\
 & =G\mathcal{L}^{\prime}G-\Delta\nonumber \\
 & =G(\mathcal{L}^{\prime}-\Delta)G,
\end{align}
which implies that $\mathcal{L}^{\boldsymbol{k}}$ is similar to $\mathcal{L}^{\prime}-\Delta$
and as such share the same eigenvalues. Since $\mathcal{G}$ is strongly
connected, $\mathcal{L}^{\prime}$ is irreducible. Also note that
$l_{ii}^{\prime}>0$ for all $i\in\underline{n}$. Let $u_{0}=\underset{i\in\underline{n}}{\max}\left\{ l_{ii}^{\prime}\right\} $
and denote $T=-\mathcal{L}^{\prime}+u_{0}I$. Then $T$ is an irreducible
non-negative matrix. 

Let $\boldsymbol{v}$ be the eigenvector of $\mathcal{L}^{\prime}$
corresponding to the zero eigenvalue, namely, $\mathcal{L}^{\prime}\boldsymbol{v}=0$.
Then ${\color{black}T\boldsymbol{v}=-\mathcal{L}^{\prime}\boldsymbol{v}+u_{0}I\boldsymbol{v}=u_{0}\boldsymbol{v}}$
and $\boldsymbol{v}$ is the eigenvector of $T$ corresponding to
the eigenvalue $u_{0}$. Since $[T]_{ii}=u_{0}-l_{ii}^{\prime}$ for
all $i\in\underline{n}$ and $[T]_{ij}=-l_{ij}^{\prime}$ for all
$i\not=j\in\underline{n}$, every eigenvalue of $T$ lies within at
least one of the following Gershgorin discs 
\begin{equation}
\left\{ z\in\mathbb{C}\,:\thinspace\mid z-u_{0}+l_{ii}^{\prime}\mid\leq\sum_{j\neq i}\mid l_{ij}^{\prime}\mid\right\} 
\end{equation}
where $i\in\underline{n}$. Hence, if $\lambda\in\boldsymbol{\lambda}(T)$,
then $|\lambda|\leq u_{0}$. Thus, $\rho(T)=u_{0}$. Moreover, according
to Perron--Frobenius theorem (\cite[Theorem 8.4.4,  p. 534]{horn1990matrix}),
$u_{0}$ is a simple eigenvalue of $T$. We now have 
\begin{equation}
T+\Delta=-(\mathcal{L}^{\prime}-\Delta)+u_{0}I\geq T,
\end{equation}
which implies that $\rho(T+\Delta)\geq\rho(T)$. 

If $\rho(T+\Delta)=\rho(T)$, then for every $j\in\underline{n}$,
$[T+\Delta]_{jj}$ and $[T]_{jj}$ must have the same modulus according
to Wielandt's theorem (\cite[Theorem 8.4.5,  p. 534]{horn1990matrix}).
However $[T+\Delta]_{11}>[T]_{11},$ thus $\rho(T+\Delta)>\rho(T)=u_{0}.$
Let 
\begin{equation}
\rho(T+\Delta)=\rho(T)+\tau=u_{0}+\tau.
\end{equation}
Since $T+\Delta$ is an irreducible non-negative matrix, $u_{0}+\tau$
is an eigenvalue of $T+\Delta$; and denote its corresponding eigenvector
as $\boldsymbol{w}$. Then 
\begin{equation}
(u_{0}I-(\mathcal{L}^{\prime}-\Delta))\boldsymbol{w}=(u_{0}+\tau)\boldsymbol{w},
\end{equation}
implying that $(\mathcal{L}^{\prime}-\Delta)\boldsymbol{w}=-\tau\boldsymbol{w}.$
As such $-\tau$ is an eigenvalue of both $\mathcal{L}^{\prime}-\Delta$
and $\mathcal{L}^{\boldsymbol{k}}$, which is a contradiction. 

2) In spirit, the proof is similar to the case of structurally balanced,
connected, undirected signed networks, we shall omit it for space.

3) The proof is similar to the Case 3 in the proof of Theorem \ref{thm:SB-epsilon-gap-sufficient-necessary}.
\end{proof}
In parallel, we also have the following corollary.
\begin{cor}
\label{cor:SB-epsilon-gap-sufficient-necessary-directed}Let $\mathcal{G}=(\mathcal{V},\mathcal{E},W)$
be a structurally balanced, strongly connected, directed signed network.
Then, the multi-agent system (\ref{eq:feedback-overall}) achieve
the bipartite consensus for any initial state $\boldsymbol{x}(0)\in\mathbb{R}^{n}$
if and only if $\boldsymbol{k}=\boldsymbol{\delta}$. Moreover, $\underset{t\rightarrow\infty}{\text{{\bf lim}}}\boldsymbol{x}(t)=G\mathds{1}_{n}\boldsymbol{p}^{T}G\boldsymbol{x}(0)$,
where $G$ is the Gauge transformation associated with $\mathcal{G}$
and $\boldsymbol{p}^{T}GL^{\boldsymbol{\delta}}G=\boldsymbol{0}$
and $\boldsymbol{p}^{T}\mathds{1}_{n}=1$.
\end{cor}
\begin{rem}
The proof of Theorem \ref{cor:SB-epsilon-gap-sufficient-necessary-directed}
implies the steady-state of the multi-agent system (\ref{eq:feedback-overall})
from the respect of the choice of the compensation $\boldsymbol{k}$,
i.e., $\boldsymbol{k}\leq\boldsymbol{\delta}$, $\boldsymbol{k}=\boldsymbol{\delta}$
and $\boldsymbol{k}\geq\boldsymbol{\delta}$. Similar to the undirected
case, it is shown that if $\boldsymbol{k}\leq\boldsymbol{\delta}$,
the multi-agent system (\ref{eq:feedback-overall}) is unstable; if
$\boldsymbol{k}\geq\boldsymbol{\delta}$, the multi-agent system (\ref{eq:feedback-overall})
achieve the trivial consensus; if $\boldsymbol{k}=\boldsymbol{\delta}$,
the multi-agent system (\ref{eq:feedback-overall}) achieve the bipartite
consensus. 
\end{rem}
\begin{thm}
Let $\mathcal{G}=(\mathcal{V},\mathcal{E},W)$ be a strongly connected,
structurally imbalanced, weight balanced, directed signed network.
Then, there exists a compensation vector $\boldsymbol{k}^{\prime}\in\mathbb{R}^{n}$
and $\boldsymbol{k}^{\prime}\le\boldsymbol{\delta}$ such that $-\mathcal{L}^{\boldsymbol{k}^{\prime}}$
is stable.
\end{thm}
\begin{proof}
Note that as $\mathcal{G}$ is weight balanced, $\boldsymbol{\delta}(\mathcal{L}^{\prime})=\boldsymbol{\delta}(\mathcal{L})$.
For $\mathcal{L}^{\prime}=\frac{1}{2}(\mathcal{L}+\mathcal{L}^{T})$,
according to Theorem \ref{prop:SB-2}, there exists a compensation
vector $\boldsymbol{k}^{\prime}\le\boldsymbol{\delta}(\mathcal{L}^{\prime})$
such that $\mathcal{L}^{\prime}+\text{{\bf diag}}\{\boldsymbol{k}^{\prime}\}$
is positive stable. Let $\lambda_{n}$ and $\lambda_{1}$ be the largest
and smallest eigenvalues of $\mathcal{L}^{\prime}+\text{{\bf diag}}\{\boldsymbol{k}^{\prime}\}$,
respectively. Note that
\begin{align}
\mathcal{L}^{\prime}+\text{{\bf diag}}\{\boldsymbol{k}^{\prime}\} & =\frac{1}{2}((\mathcal{L}+\text{{\bf diag}}\{\boldsymbol{k}^{\prime}\})\nonumber \\
 & +(\mathcal{L}+\text{{\bf diag}}\{\boldsymbol{k}^{\prime}\})^{T}).
\end{align}
Then, $\text{{\bf Re}}(\lambda_{i}(\mathcal{L}+\text{{\bf diag}}\{\boldsymbol{k}^{\prime}\}))\in[\lambda_{1},\lambda_{n}]$
where $i\in\underline{n}$. Therefore, there exists a compensation
vector $\boldsymbol{k}^{\prime}\in\mathbb{R}^{n}$ with $\boldsymbol{k}^{\prime}\le\boldsymbol{\delta}$,
such that $-\mathcal{L}^{\boldsymbol{k}^{\prime}}$ is stable.
\end{proof}
For the structurally imbalanced strongly connected signed network,
similar to the undirected network case, it is shown that the multi-agent
system (\ref{eq:feedback-overall}) achieves trivial consensus if
$\boldsymbol{k}=\boldsymbol{\delta}$ \cite{altafini2013consensus}.
Similarly, the compensation vector can be less than $\boldsymbol{\delta}$
in order to make the system (\ref{eq:feedback-overall}) stable, as
show in the following example. 
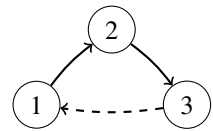
\begin{figure}[tbh]
\begin{centering}
\begin{tikzpicture}[scale=1]
	
	\node (n1) at (0,0) [circle,draw] {1};
    \node (n2) at (1,1) [circle,draw] {2};
	\node (n3) at (2,0) [circle,draw] {3};

	\path[]

	(n2) [->,thick,  bend left=8] edge node[right] {} (n3);

	\path[]
	(n1) [->,thick,  bend left=8] edge node[left] {} (n2);

	\path[]
	(n3) [->,thick, dashed,  bend left=8] edge node[left] {} (n1);

\end{tikzpicture}
\par\end{centering}
\caption{A strongly connected, directed structurally unbalanced signed network.}
\label{fig:example2}
\end{figure}

\begin{example}
Consider the signed Laplacian of a strongly connected, directed structurally
unbalanced signed network in Figure \ref{fig:example2},
\[
L=\left[\begin{array}{ccc}
-1 & 0 & 1\\
-1 & 1 & 0\\
0 & -1 & 1
\end{array}\right].
\]
Choose $\boldsymbol{k}_{0}=[1.9,0,0]^{T}<[2,0,0]^{T}=\boldsymbol{\delta}$
and compute the eigenvalues of $\mathcal{L}^{\boldsymbol{k}_{0}}$
yields $\lambda_{1}(\mathcal{L}^{\boldsymbol{k}_{0}})=0.3210+0.8392i$,
$\lambda_{2}(\mathcal{L}^{\boldsymbol{k}_{0}})=0.3210-0.8392i$, $\lambda_{3}(\mathcal{L}^{\boldsymbol{k}_{0}})=1.8581$,
namely, $-\mathcal{L}^{\boldsymbol{k}_{0}}$ is stable and $\boldsymbol{k}_{0}<\boldsymbol{\delta}$.
\end{example}
\begin{thm}
\label{thm:thm5}Let $\mathcal{G}=(\mathcal{V},\mathcal{E},W)$ be
a structurally imbalanced, strongly connected signed network. If the
smallest eigenvalue of $\mathcal{L}^{\boldsymbol{\delta}}$ (ordered
by real parts) is real and simple, denoted by $\lambda_{1}(\mathcal{L}^{\boldsymbol{\delta}})$
with left and right eigenvectors $\boldsymbol{v}_{1}^{l}$ and $\boldsymbol{v}_{1}^{r}$
satisfying $(\boldsymbol{v}_{1}^{l})^{T}\boldsymbol{v}_{1}^{r}=1$.
Then, there exists a compensation vector $\boldsymbol{k}=\boldsymbol{\delta}-\lambda_{1}(\mathcal{L}^{\boldsymbol{\delta}})\mathds{1}_{n}$
such that eigenvalues of $\mathcal{L}^{\boldsymbol{k}}$ have non-negative
real parts and $\mathcal{L}^{\boldsymbol{k}}$ has simple zero eigenvalue
with left and right eigenvectors $\boldsymbol{v}_{1}^{l}$ and $\boldsymbol{v}_{1}^{r}$.
Moreover, the compensated multi-agent network (\ref{eq:feedback-overall})
achieve cluster consensus characterized by 
\[
\underset{t\rightarrow\infty}{\text{{\bf lim}}}\boldsymbol{x}(t)=\boldsymbol{v}_{1}^{r}(\boldsymbol{v}_{1}^{l})^{T}\boldsymbol{x}(0).
\]
\end{thm}
\begin{proof}
The proof is similar to the proof of Theorem \ref{prop:imbalanced-undirected},
we omitted here for brevity.
\end{proof}
We provide an example to illustrate the steady-state of compensated
dynamics (\ref{eq:feedback-overall}) on directed, structurally imbalanced
signed networks. 
\begin{figure}[tbh]
\begin{centering}
\includegraphics[width=9cm]{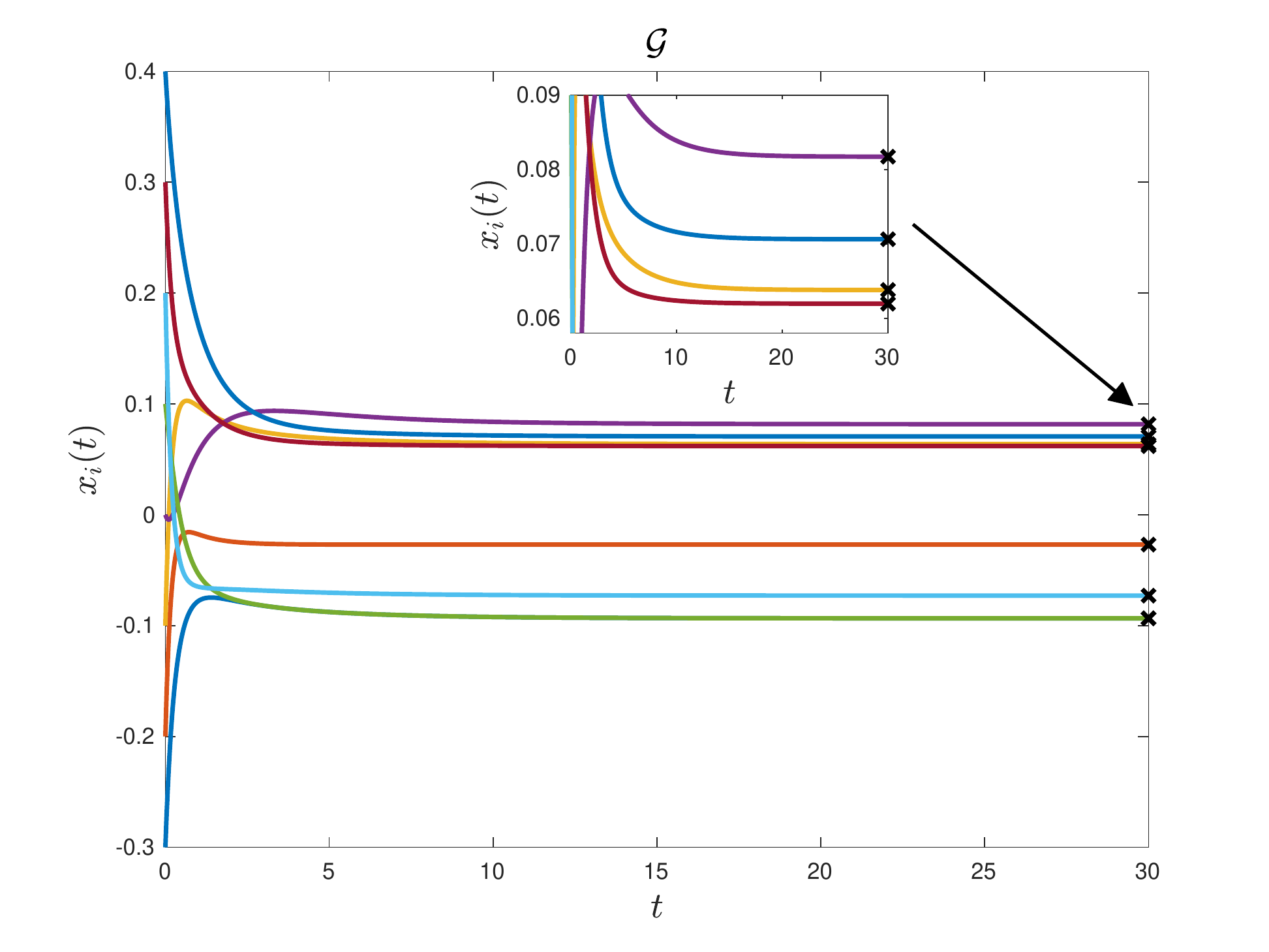}
\par\end{centering}
\caption{The state trajectory of compensated dynamics for the structurally
imbalanced network $\mathcal{G}_{2}$ in Figure \ref{fig:8-node-signed-network}.}
\label{fig:x-f-G2}
\end{figure}

Consider the directed, structurally imbalanced, signed network $\mathcal{G}_{2}$
in Figure \ref{fig:8-node-signed-network}. The steady-state can be
predicted via $\boldsymbol{v}_{1}^{l}$ and $\boldsymbol{v}_{1}^{r}$,
the left and right normalized eigenvector associated with $\lambda_{1}(\mathcal{L}^{\boldsymbol{\delta}}(\mathcal{G}_{2}))$,
namely,
\begin{equation}
\underset{t\rightarrow\infty}{\text{{\bf lim}}}\boldsymbol{x}(t)=\boldsymbol{v}_{1}^{r}(\boldsymbol{v}_{1}^{l})^{T}\boldsymbol{x}(0).\label{eq:x_f_G2_prediction}
\end{equation}
The state trajectory of the compensated dynamics which is shown in
Figure \ref{fig:x-f-G2}, where the black crosses indicated the steady-state
predicted by (\ref{eq:x_f_G2_prediction}). In this example, the right
and left normalized eigenvector associated with $\lambda_{1}(\mathcal{L}^{\boldsymbol{\delta}}(\mathcal{G}_{2}))$
are
\begin{align}
\boldsymbol{v}_{1}^{r} & =(-0.45,\,-0.13,\,0.31,\,0.39,\nonumber \\
 & -0.45,\,-0.35,\,0.30,\,0.34)^{T},
\end{align}
and
\begin{align}
\boldsymbol{v}_{1}^{l} & =(-0.50,\,-0.31,\,0.29,\,0.38,\nonumber \\
 & -0.50,\,-0.39,\,0.17,\,0.26)^{T},
\end{align}
respectively. The initial state of agents is
\begin{equation}
\boldsymbol{x}(0)=(-0.4,\,-0.3,\,-0.2,\,-0.1,\,0,\,0.1,\,0.2,\,0.3)^{T}.
\end{equation}
Therefore, one can get that
\begin{align}
\underset{t\rightarrow\infty}{\text{{\bf lim}}}\boldsymbol{x}(t) & =(-0.12,\,-0.03,\,0.08,\,0.11,\nonumber \\
 & \,-0.12,\,-0.09,\,0.08,\,0.09)^{T}.
\end{align}

Finally, we discuss whether or not the compensation vector $\boldsymbol{k}$
can be less than $\boldsymbol{\delta}$ so as to render the system
(\ref{eq:feedback-overall}) achieving the trivial consensus.
\begin{table*}[t]
\begin{centering}
\begin{tabular}{|c|c|c|c|c|c|c|}
\hline 
$\boldsymbol{k}>\boldsymbol{\delta}$ & \multicolumn{6}{c|}{Trivial consensus}\tabularnewline
\hline 
\multirow{2}{*}{$\boldsymbol{k}=\boldsymbol{\delta}$} & \multicolumn{3}{c|}{Structurally balanced} & \multicolumn{3}{c|}{Structurally imbalanced}\tabularnewline
\cline{2-7} \cline{3-7} \cline{4-7} \cline{5-7} \cline{6-7} \cline{7-7} 
 & \multicolumn{3}{c|}{${\displaystyle \frac{1}{n}G\boldsymbol{1}_{n}(\boldsymbol{v}_{1}^{l})^{T}G\boldsymbol{x}(0)}$} & \multicolumn{3}{c|}{Trivial consensus}\tabularnewline
\hline 
\multirow{3}{*}{$\boldsymbol{k}<\boldsymbol{\delta}$} & \multicolumn{3}{c|}{Undirected networks} & \multicolumn{3}{c|}{Directed networks}\tabularnewline
\cline{2-7} \cline{3-7} \cline{4-7} \cline{5-7} \cline{6-7} \cline{7-7} 
 & Structurally balanced & \multicolumn{2}{c|}{Structurally imbalanced} & Structurally balanced & \multicolumn{2}{c|}{Structurally imbalanced}\tabularnewline
\cline{2-7} \cline{3-7} \cline{4-7} \cline{5-7} \cline{6-7} \cline{7-7} 
 & Unstable & $\boldsymbol{k}=\boldsymbol{\delta}-\lambda_{1}(\mathcal{L}^{\boldsymbol{\delta}})\mathds{1}_{n}$ & ${\displaystyle \sum_{j=1}^{p}}\boldsymbol{v}_{j}\boldsymbol{v}_{j}^{T}\boldsymbol{x}(0)$ & Unstable & $\boldsymbol{k}=\boldsymbol{\delta}-\lambda_{1}(\mathcal{L}^{\boldsymbol{\delta}})\mathds{1}_{n}$ & $\boldsymbol{v}_{1}^{r}(\boldsymbol{v}_{1}^{l})^{T}\boldsymbol{x}(0)$\tabularnewline
\hline 
$\boldsymbol{k}\veebar\boldsymbol{\delta}$ & \multicolumn{6}{c|}{No bipartite consensus}\tabularnewline
\hline 
\end{tabular}
\par\end{centering}
\caption{Selection of compensation vector $\boldsymbol{k}$ and the collective
behavior of compensated dynamics (\ref{eq:feedback-overall}) ($\boldsymbol{x}(\infty)$).}
\label{tbl}
\end{table*}

\begin{cor}
Let $\mathcal{G}=(\mathcal{V},\mathcal{E},W)$ be a structurally imbalanced,
weight balanced, connected signed network. If the compensation vector
satisfies
\begin{equation}
\boldsymbol{k}>\boldsymbol{\delta}-\lambda_{1}\left(\frac{1}{2}(\mathcal{L}^{\boldsymbol{\delta}}+(\mathcal{L}^{\boldsymbol{\delta}})^{T})\right)\mathds{1}_{n},
\end{equation}
then, the multi-agent system (\ref{eq:feedback-overall}) achieve
the trivial consensus.
\end{cor}
\begin{proof}
Note that,
\begin{align}
\frac{1}{2}((\mathcal{L}+\text{{\bf diag}}\{\boldsymbol{k}\}) & +(\mathcal{L}+\text{{\bf diag}}\{\boldsymbol{k}\})^{T})\nonumber \\
 & =\frac{1}{2}(\mathcal{L}+\mathcal{L}^{T})+\text{{\bf diag}}\{\boldsymbol{k}\},
\end{align}
since $\mathcal{G}$ is weight balanced, then $\frac{1}{2}(\mathcal{L}+\mathcal{L}^{T})$
is the Laplacian matrix of an undirected network corresponding to
directed network $\mathcal{G}$. According to Corollary \ref{cor:imbalanced-undirected-two},
if 
\begin{equation}
\boldsymbol{k}>\boldsymbol{\delta}-\lambda_{1}\left(\frac{1}{2}(\mathcal{L}^{\boldsymbol{\delta}}+(\mathcal{L}^{\boldsymbol{\delta}})^{T})\right)\mathds{1}_{n},
\end{equation}
then $\frac{1}{2}(\mathcal{L}+\mathcal{L}^{T})+\text{{\bf diag}}\{\boldsymbol{k}\}$
is positive definite, i.e., all the eigenvalues of $\frac{1}{2}(\mathcal{L}+\mathcal{L}^{T})+\text{{\bf diag}}\{\boldsymbol{k}\})$
are positive. Due to the fact
\begin{equation}
\text{{\bf Re}}(\lambda_{1}(\mathcal{L}+\text{{\bf diag}}\{\boldsymbol{k}\}))\geq\lambda_{1}\left(\frac{1}{2}(\mathcal{L}+\mathcal{L}^{T})+\text{{\bf diag}}\{\boldsymbol{k}\}\right),
\end{equation}
therefore, the multi-agent system (\ref{eq:feedback-overall}) achieve
the trivial consensus.
\end{proof}
Now, one can summarize the collective behavior of compensated dynamics
(\ref{eq:feedback-overall}) in terms of compensation vector $\boldsymbol{k}$
in Table \ref{tbl}.

\section{Compensated Dynamics and Eventually Positivity }

In literature, eventually positivity of a matrix $M\in\mathbb{R}^{n\times n}$
turns out to be relevant in the characterization of the positive semidefiniteness
of signed Laplacian from an algebraic perspective \cite{altafini2019investigating,chen2020spectral,altafini2014predictable}.
In this section, we shall further discuss the correlation between
the stability of the compensated dynamics obtained by self-loop compensation
and eventually positivity. We first recall some basic concepts regarding
the eventually positivity of matrices.
\begin{defn}
A matrix $M\in\mathbb{R}^{n\times n}$ is said to be eventually positive
(exponentially positive) if there exists a $k_{0}\in\mathbb{N}$,
such that $M^{k}>0$ ($e^{Mk}>0$) for all $k\geq k_{0}$. 
\end{defn}
The relationship between eventual positivity and eventual exponential
positivity is characterized by the following lemma. 
\begin{lem}
A matrix $M\in\mathbb{R}^{n\times n}$ is eventually exponentially
positive if and only if there exists $s\geq0$, such that $sI+M$
is eventually positive. 
\end{lem}
\begin{lem}
\label{lem:eventually-positivity} Let $\mathcal{G}=(\mathcal{V},\mathcal{E},W)$
be a connected, undirected signed network. Then the signed Laplacian
$\mathcal{L}(\mathcal{G})$ is positive semi-definite with only simple
zero eigenvalue if and only if $-\mathcal{L}(\mathcal{G})$ is eventually
exponentially positive.
\end{lem}
\begin{lem}
Let $\mathcal{G}=(\mathcal{V},\mathcal{E},W)$ be a weight balanced,
connected directed signed network with the corresponding Laplacian
$\mathcal{L}(\mathcal{G})$. Then $-\mathcal{L}(\mathcal{G})$ is
eventually exponentially positive if and only if $-\mathcal{L}(\mathcal{G})$
is marginally stable.
\end{lem}
With the above statements, we have the following result characterizing
$\mathcal{L}^{\boldsymbol{k}}$ for undirected signed networks with
eventually (exponentially) positivity.
\begin{thm}
\label{thm:EP-undirected}Let $\mathcal{G}=(\mathcal{V},\mathcal{E},W)$
be a structurally balanced, connected, undirected signed network.
Then, the following statements hold.

1) If $\boldsymbol{k}\leq\boldsymbol{\delta}$ and $\boldsymbol{k}\not=\boldsymbol{\delta}$,
then $-\mathcal{L}^{\boldsymbol{k}}$ is not eventually exponentially
positive;

2) If $\boldsymbol{k}\ge\boldsymbol{\delta}$ and $\boldsymbol{k}\not=\boldsymbol{\delta}$,
then then $-\mathcal{L}^{\boldsymbol{k}}$ is not eventually exponentially
positive.
\end{thm}
\begin{proof}
According to Theorem \ref{thm:SB-epsilon-gap-sufficient-necessary},
one can see that for the case $\boldsymbol{k}\leq\boldsymbol{\delta}$($\boldsymbol{k}\not=\boldsymbol{\delta}$)
or $\boldsymbol{k}\ge\boldsymbol{\delta}$ ($\boldsymbol{k}\not=\boldsymbol{\delta}$),
$\mathcal{L}^{\boldsymbol{k}}$ has at least one negative eigenvalue
or is positive definite. However, from Lemma \ref{lem:eventually-positivity},
$-\mathcal{L}^{\boldsymbol{k}}$ is eventually exponentially positive
which is equivalent to that $\mathcal{L}^{\boldsymbol{k}}$ is positive
semi-definite with one zero eigenvalue. Then, the proof is finished.
\end{proof}
In parallel, the similar result for the case of directed signed networks
is as follows.
\begin{thm}
\label{thm:EP-directed}Let $\mathcal{G}=(\mathcal{V},\mathcal{E},W)$
be a weight balanced, connected directed signed network. Then, the
following statements hold.

1) If $\boldsymbol{k}\leq\boldsymbol{\delta}$ and $\boldsymbol{k}\not=\boldsymbol{\delta}$,
then $-\mathcal{L}^{\boldsymbol{k}}$ is not eventually exponentially
positive;

2) If $\boldsymbol{k}\ge\boldsymbol{\delta}$ and $\boldsymbol{k}\not=\boldsymbol{\delta}$,
then $-\mathcal{L}^{\boldsymbol{k}}$ is not eventually exponentially
positive.
\end{thm}
\begin{proof}
The proof follows similar procedures as that in Theorem \ref{thm:SB-epsilon-gap-sufficient-necessary},
we shall omit it for space.
\end{proof}
According to Theorems \ref{thm:EP-undirected} and \ref{thm:EP-directed},
one can see that the eventually (exponentially) positivity cannot
characterize the stability of compensated dynamics (\ref{eq:feedback-overall})
under undirected/directed signed networks. This gap is eventually
filled by the results in this paper (Theorems \ref{thm:SB-epsilon-gap-sufficient-necessary}
and \ref{thm:SB-epsilon-gap-sufficient-necessary-directed}).

\section{Conclusion Remarks}

The negative edges in a signed network can often lead to an unstable
Laplacian matrix, essentially owing to the lost of its diagonal dominance.
In this paper, a graph-theoretic characterization of stability of
signed Laplacian is provided from the perspective of negative cut
set\emph{ }which not only provides insight into the design of stable
signed networks but also indicates that manipulation of edge weights
cannot change an unstable signed network containing negative cut set
into a stable one. A self-loop compensation mechanism is subsequently
introduced and the connection between self-loop compensation (for
re-establishment of diagonal dominance of signed Laplacian) and stability/consensus
of signed network is examined. The correlation between the selection
of compensation vector and collective behavior of the network is shown
to be closely related to the structural balance of the underlying
signed network. The correlation between the stability of the compensated
dynamics obtained by self-loop compensation and eventually positivity
is further discussed. The results in this work eventually provide
a novel perspective on the stability of signed Laplacian and its correspondence
with the graph-theoretic characterization of the underlying signed
network. 

\bibliographystyle{IEEEtran}
\phantomsection\addcontentsline{toc}{section}{\refname}\bibliography{mybib}

\end{document}